\documentclass{amsart}

\usepackage{latexsym}
\usepackage{linearb}
\usepackage[LGR, OT1]{fontenc}
\usepackage{amsmath,amsthm}
\usepackage{amssymb}
\usepackage{upgreek}
\usepackage[greek,english]{babel}
\usepackage{teubner}
\usepackage{MnSymbol, wasysym}
\usepackage{graphicx}
\usepackage{marvosym}
\usepackage{eufrak}
\usepackage{color}
\def\f12{\frac 1 2}

\def\a{\alpha}
\def\b{\beta}

\def\de{\delta}

\def\la{\lambda}
\def\La{\Lambda}

\def\pa{\partial}

\def\nab{\nabla}
\def\pa{\partial}

\def\f12{\frac 1 2}

\newcommand{\nabb}{\mbox{$\nabla \mkern-13mu /$\,}}

\newcommand\lessflat{{{\mbox{$\flat  \mkern-12mu {}^{\_}$}}}}

\newtheorem{remark}{Remark}[section]
\newtheorem{lemma}{Lemma}[subsection]
\newtheorem{theorem}{Theorem}
\newtheorem*{theorem*}{Theorem}
\newtheorem{proposition}{Proposition}[subsection]

\newtheorem{corollary}{Corollary}

\title{The black hole stability problem for linear scalar perturbations} 
\author{Mihalis Dafermos}
\address{University of Cambridge,
Department of Pure Mathematics and Mathematical Statistics,
Wilberforce Road, Cambridge CB3 0WB United Kingdom}
\author{Igor Rodnianski}
\address{Princeton University,
Department of Mathematics, Fine Hall, Washington Road,
Princeton, NJ 08544 United States} 

\date\today
\begin{document}
\maketitle

\begin{abstract}
We review our recent work on linear stability for scalar perturbations of
Kerr spacetimes, that is to say, boundedness and decay properties for
solutions of the scalar wave equation $\Box_g\psi=0$
on Kerr 
exterior backgrounds $(\mathcal{M},g_{a,M}$). We begin with
the very slowly rotating case $|a|\ll M$,  where first
boundedness and then decay has been shown in rapid developments
over the last
two years, following earlier progress in the Schwarzschild case $a=0$.
We then turn to the  general
subextremal range $|a|<M$, where we give here for the first time
the essential elements of a proof of definitive decay bounds
for solutions $\psi$. These developments give hope that the problem
of the non-linear stability of the Kerr family of black holes might soon be addressed.
This paper accompanies a talk by one of the authors (I.R.) at
the  12th Marcel Grossmann Meeting, Paris, June 2009.
\end{abstract}

\section{Introduction}
The stability problem for 
the 
Kerr family $(\mathcal{M},g_{a,M})$ of black hole  solutions to
the Einstein vacuum equations
\begin{equation}\label{eq:einst}
R_{\a\b}(g) =0,
\end{equation}
the system of nonlinear hyperbolic equations governing general relativity,
is one 
of the most important unresolved issues in the theory.

The ultimate goal
would be to understand nothing less than the dynamical stability of Kerr 
as a family of solutions
to the Cauchy problem for $(\ref{eq:einst})$, in the spirit of the monumental monograph of 
Christodoulou--Klainerman~\cite{ck} (see also~\cite{lr}) proving the
non-linear stability of Minkowski space.  
For the time being, however, essentially all work in the black hole case has been
confined to the linearised setting.
The simplest such linear problem
is that of \emph{scalar} perturbations, i.e.~the study of
solutions $\psi$ of
\begin{equation}\label{eq:box}
\Box_g\psi=0,
\end{equation}
on a fixed Kerr background $(\mathcal{M},g=g_{a,M})$. 
Equation $(\ref{eq:box})$ can be thought
of as a poor man's substitute for the more complicated problem of 
\emph{gravitational} perturbations,
obtained by linearising \eqref{eq:einst} 
around the Kerr family.

The present paper will
 review the problem of scalar perturbations~$(\ref{eq:box})$, in particular,
the rapid recent progress concerning the Schwarzschild $a=0$ and
very slowly rotating $|a|\ll M$ cases. We shall then provide,
for the first time here, 
the essential elements of a proof of linear stability for 
scalar perturbations in the 
general subextremal case $|a|<M$, that is to say, boundedness and 
decay statements (see Theorem~\ref{thrm:state} and its corollaries) 
for general solutions $\psi$ of $(\ref{eq:box})$ with $g=g_{a,M}$
arising from arbitrary finite-energy initial data prescribed on a Cauchy surface.

In the present introductory section, we shall attempt a self-contained
overview, which will be fleshed out in Sections~\ref{KST}--\ref{revisit}.
Sections \ref{ayrilik}--\ref{sec:large}  will then give technical details of the proof
of Theorem~\ref{thrm:state}  in the
general $|a|<M$ case.

As there is much confusion concerning the state of affairs in this subject, we begin
with a brief discussion of the classical literature.

\subsection{The classical mode analysis}
\label{classical}
In the classical analysis, one does not study general solutions of $(\ref{eq:box})$
but rather individual \emph{modes}.
In the Schwarzschild case $a=0$, study of this problem was initiated in~\cite{regge}. 
A mode is a solution of $(\ref{eq:box})$ of the form
\begin{equation}
\label{amode}
\psi^{\omega}_{m\ell} (r) Y_{m\ell} (\theta,\phi) e^{i\omega t}
\end{equation}
where $Y_{m\ell}$ denote the standard spherical harmonics, and $(t,r,\theta,\phi)$
denote standard Schwarzschild coordinates.
In view of the stationarity and spherical symmetry of Schwarzschild, from $(\ref{eq:box})$
and the ansatz $(\ref{amode})$ one derives an ODE for 
$\psi^{\omega}_{m\ell} (r)$ in $r$,
which can be rewritten in the form
\[
u''+\omega^2u=Vu
\]
where $V$ is a nonnegative potential function, depending on $m$ and $\ell$.
(Remarkably, this is also the case for gravitational perturbations~\cite{vish}.)
It is immediate that there are no solutions of the form $(\ref{amode})$
with finite energy at $t=0$ and ${\rm Im} (\omega)>0$. This is the statement
of `mode stability' of Schwarzschild. 

The literature on mode solutions $(\ref{amode})$ for Schwarzschild is now vast and
addresses questions such as 
understanding transmission and reflection coefficients (relevant
for black hole scattering) and  locating the complex frequencies $\omega$
of so-called \emph{quasinormal modes}
(solutions of type $(\ref{amode})$ with ${\rm Im}(\omega)<0$ and finite
energy on `hyperboloidal' slices). These 
have been explored numerically and by methods of asymptotic analysis.
See~\cite{chandra}.

The existence of mode-type solutions for $(\ref{eq:box})$ on Schwarzschild--the 
essential starting
point for this classical analysis--follows
from the dimensionality of its Lie algebra of symmetries generated by
the stationary Killing field and the rotations.
Remarkably, as discovered by Carter, although the Kerr metric for $a\ne 0$ only has
a $2$-dimensional algebra of Killing fields, the wave equation $(\ref{eq:box})$
on Kerr still admits a non-trivial complete separation which
allows one again to define mode-type solutions  of the form
\begin{equation}
\label{modetype}
\psi^{\omega}_{m\ell} (r) S_{m\ell}(a\omega, \phi) e^{im\phi} e^{-i\omega t}
\end{equation}
so that $\psi^{\omega}_{m\ell} (r)$ again satisfies an ODE that can be rewritten as
\[
u''+\omega^2u-V(a\omega)u=0,
\]
where the potential $V$ now depends, in addition to the labels $(m, \ell)$, on the
frequency $\omega$. Here, $(t,r,\theta,\phi)$ denote 
Boyer-Lindquist coordinates.
The $S_{m\ell}$ are not explict but can be characterized 
as the eigenvectors of an associated operator depending on $a\omega$, which reduces to the spherical Laplacian if $a\omega=0$
(see Section~\ref{ayrilik}).
The above separation is related to the complete integrability of geodesic flow,
discovered previously by Carter, and in fact has its origin in 
the existence of an additional ``hidden'' symmetry of the Kerr metric~\cite{walker}. 
This separation
property was extended to gravitational perturbations by Teukolsky~\cite{teuk}.

For the mode analysis of scalar perturbations, 
the main new feature which emerges
in passing from Schwarzschild to Kerr is the problem of \emph{superradiance},
first discussed by Zeldovich~\cite{zeld}.
The geometric origin of this phenomenon is simply that  
the `stationary' vector field $\partial_t$ fails to be timelike (see Section~\ref{ergosec}), 
and thus the energy flux
associated to this vector field is no longer nonnegative (see Section~\ref{supersec}).
At the level of mode analysis for $\omega\in\mathbb R$,
this difficulty reflects itself only in the range
\begin{equation}
\label{suprange}
0\le m\omega/a    <  \frac{m^2}{2M(M+\sqrt{M^2-a^2})}.
\end{equation}
It is precisely in the frequency range $(\ref{suprange})$ that the
sign of the energy flux
 through the event
horizon of a solution of the form $(\ref{modetype})$ is negative. The transmission and reflection coefficients
of superradiant modes 
were probed by Starobinsky~\cite{star}, who derived bounds
in various asymptotic regimes. A priori, superradiance might allow
for the existence of initially finite-energy solutions of type $(\ref{modetype})$ with 
${\rm Im}(\omega)>0$, and thus, growing exponentially in time. 
These are excluded (for both the scalar and gravitational case) in a
remarkable paper of Whiting~\cite{whiting}, who thus shows
`mode stability' of the Kerr family. Whiting's result can be viewed as the culmination
of the  classical mode analysis of this problem.

\subsection{Linear stability?}
\label{LS?}
Despite the successes just outlined, the above classical mode analysis was not able
to resolve what is perhaps the most basic question of linear theory, 
namely that of \emph{linear stability}, 
even in the case of scalar perturbations $(\ref{eq:box})$.  
Indeed, the statement of `mode stability' is still completely consistent with the statement
that general linear perturbations $\psi$ with finite initial energy
grow unboundedly in time
\begin{equation}
\label{nothap}
\limsup_{t\to\infty} \psi(t,r,\theta,\phi) =\infty.
\end{equation}

The reason for the failure of mode analysis to exclude $(\ref{nothap})$
can be summarised by the following:
\begin{itemize}
\item
It is not a priori clear that solutions of $(\ref{eq:box})$ with finite energy 
can be represented
as a superposition of modes with $\omega\in \mathbb R$.
\item
Even if solutions can be so represented, 
statements at the level of individual modes typically do not
imply statements for the superposition of infinitely many modes.
\end{itemize}

Indeed, in trying to `sum' information at the level of individual modes, one 
has to face certain fundamental
physical obstructions
which are in fact common to all 
wave equations in more than one spatial dimension and 
necessarily (see~\cite{Littman}) 
lead one to consider \emph{energy-type quantities}. It is only such quantities that will identify
what it is that bounds $\psi$ and excludes $(\ref{nothap})$.
In the black hole case,
these obstructions force us to revisit 
to what extent we actually understand superradiance,
the red-shift effect and the role of solutions concentrated near
trapped null geodesics, two issues which we shall discuss at length later in this paper. 
As we shall see, it is  precisely the interaction of these difficulties which gives the present 
true linear stability problem its characteristic
flavour.

\subsection{Schwarzschild}
A prototype for a true linear stability result  is provided
by the celebrated theorems of Wald~\cite{drimos} and
Kay--Wald \cite{kw}, which show that solutions $\psi$
of \eqref{eq:box} on a Schwarzschild background remain bounded by a constant
for all time, in particular, $(\ref{nothap})$ does \emph{not} occur:
\begin{theorem*}[Kay--Wald~\cite{kw}]
Let $\psi$ be a solution of the wave equation on the Schwarzschild background
$(\mathcal{M},g=g_M)$ arising from sufficiently regular initial data on a Cauchy
hyperusrface $\Sigma$. Then
\begin{equation}
\label{thebns}
|\psi|\le C\sqrt{D}
\end{equation}
in the exterior region up to and including the horizon,
where $D$ is a suitable higher order energy quantity
evaluated on the initial hypersurface $\Sigma$.\footnote{Otherwise, simply excluding
$(\ref{nothap})$ without providing such a bound would be a `mathematical' result
with no physical interpretation.}
\end{theorem*}

The proof of the above Theorem effectively initiated the study
of the linear stability problem from the modern PDE point of view, according to which:
\begin{itemize}
\item
The most fundamental boundedness statements are at the level of quantitative control
of energy quantities.
\item
Pointwise boundedness statements follow from energy-based statements
after commuting the wave equation
with various vector fields, repeating the energy bounds, and
applying Sobolev inequalities. 
\end{itemize}
In the above, no mode-analysis need be appealed to.
Such arguments are essential in proving 
existence theorems for \emph{non-linear} wave equations, like
the Einstein equations $(\ref{eq:einst})$. 
From the point of view of~\cite{drimos,kw}, 
linear stability \emph{away} from the horizon for $(\ref{eq:box})$ essentially
thus follows from `standard arguments' in PDE and 
the fact that the stationary Killing field $\partial_t$ is timelike in the exterior,
allowing the energy method to be applied as usual.

The more novel part of the proof of the above Theorem  concerns the horizon:
Since $\partial_t$ becomes null on the event horizon, 
its associated conserved energy degenerates there, and
this does \emph{not} allow the standard method to apply to obtain bounds on $\psi$
\emph{up to and including} the horizon (see~Sections~\ref{ergosec} and~\ref{Conservative}). 
This difficulty was circumvented in~\cite{drimos,kw}
by a number of clever, but fragile geometric arguments, exploiting in 
a non-trivial way the spherical symmetry
but also the discrete symmetry of maximally extended Schwarzschild.

A simpler, and more robust, 
proof of the Kay--Wald Theorem follows by specialising~\cite{dr6}, to be discussed
below in the context of Kerr, and this new approach sheds light
on what is the physical origin of the boundedness property at the horizon: 
it is precisely the celebrated {\bf\emph{red-shift effect}},
which can be quantified by exploiting in an appropriate way
the generalised energy identities corresponding to a suitably-defined
local-observer-based energy (see Sections~\ref{RSgeo},~\ref{R-S-E} and~\ref{aside}).
This proof again does  not use any decomposition of $\psi$ into
either spherical harmonics or modes.

An additional linear stability statement supplementing
 the above Theorem would be to prove
quantitative \emph{decay} bounds for   solutions, showing not only
$(\ref{thebns})$, but that
$\psi$ can be bounded by a fixed decaying function whose modulus depends
only on the strength of initial data.
(Proving such a statement would be
essential from the point of view of the question of {\it nonlinear stability},
as decay is the only known mechanism for 
nonlinear stability (see the proof of stability of Minkowski space \cite{ck,lr}).) 
 Even in the Schwarzschild case,
such decay results have only become available in the past few years:
\cite{bluester, dr3, dr5, mmtt, luk1, toh}.
Proving these decay results has led in particular to better understanding of
how various aspects of the underlying geometry interact in the analysis
of $(\ref{eq:box})$, in particular,
quantifying the `obstruction' to decay caused by the presence of
trapped null geodesics associated to the photon sphere (see Section~\ref{TNG}).\footnote{This obstruction
is not present when one restricts to finitely many spherical harmonics, as
is often the case in the literature.}
This latter obstruction is effectively captured by use of virial-type energy
currents whose divergence is non-negative definite and degenerates
precisely on such geodesics (see Section~\ref{virialss}). 
Like the new proof~\cite{dr6} of
boundedness, the final version of the proof of decay~\cite{dr5} 
in the Schwarzschild case takes place entirely in physical space
and \emph{does not require} frequency decompositions or
 spherical harmonics. (See Sections~\ref{sec:three} and~\ref{putting}.)

\subsection{The Kerr family}
Up until two years ago, even the analogue of the Kay--Wald Theorem was
a completely
open problem for the Kerr family, in view of superradiance and the lack
of a non-negative definite conserved energy from which to start. 
The first result to appear concerning general solutions $\psi$ of~$(\ref{eq:box})$, 
see~\cite{dr6}, was a proof of the analogue of the Kay--Wald Theorem on a
large class of  underlying metrics $g$, a class which included as a special case
the very slowly rotating Kerr metrics $g_{a,M}$ with $|a|\ll M$. 
Results giving additional decay-type statements,
now specialised to the exactly Kerr case with $|a|\ll M$, 
rapidly followed in independent  works (see Chapter 5.3
 of~\cite{jnotes} and~\cite{toh}), and
later also~\cite{andB}.  We also mention previous work concerning
non-quantitative statements for azimuthal modes~\cite{fksy, fksy2}.
The state of the field up to this point is described in detail in~\cite{jnotes}.

The purpose of the present paper is to review the small $|a|\ll M$ proof
of~\cite{jnotes} and
complete the picture, by giving, for the first time here, the essential elements
of the proof of decay results for the general subextremal case $|a|<M$.
(A complete presentation, together with a proof of the 
refinement \eqref{eq:pdecay'} will be forthcoming in \cite{drf1},\cite{drf2}.)
The main result is

\begin{theorem}\label{thrm:state}
Let $\psi$ be a solution of the wave equation $\Box_g\psi=0$ on the Kerr background 
$g=g_{M,a}$ for $|a|<M$  with sufficiently regular initial data on a Cauchy 
hypersurface $\Sigma$,
and let 
$\{\Sigma_\tau\}_{\tau\ge 0}$ be the
foliation defined in Section~\ref{sec:sigma}. Then, with 
$|\pa\psi|^2$ denoting the square sum of non-degenerate derivatives 
as in Section~\ref{sec:nder}, 
we have
\begin{itemize}
\item Boundedness of energy\footnote{Note that this energy quantity is thus 
non-degenerate at the horizon.}
\begin{equation}
\int_{\Sigma_\tau}|\partial\psi|^2\le C\int_{\Sigma_0}|\partial\psi|^2
\end{equation}
\item Integrated local energy decay: for arbitrary $r_+<R_1<\infty$
\begin{equation}\label{eq:enerdloc}
\int_0^\infty \int_{\Sigma_\tau\cap \{r_+\le r\le R_1\}} \left (\chi|\pa\psi|^2+|\psi|^2\right) \le C_{R_1} 
\int_{\Sigma_0} |\pa\psi|^2
\end{equation}
and 
\begin{equation}\label{eq:enerdloc}
\int_\tau^{2\tau}\int_{\Sigma_\tau\cap \{r_+\le r\le R_1\}} |\pa\psi|^2\le C_{R_1}D \tau^{-2}.
\end{equation}
\item Polynomial-time decay of the energy-flux: 
\begin{equation}\label{eq:enerd}
\int_{\Sigma_\tau} |\pa\psi|^2\le  CD \tau^{-2}.
\end{equation}
\end{itemize}
Here, $\chi$ is a smooth cutoff function which vanishes in a neighborhood of the physical space projection of the trapped set, 
and $D$ denotes the square of an appropriate (weighted Sobolev) norm of the data,
involving higher derivatives of $\psi$. 
\end{theorem}
As described above, combining energy boundedness
and decay with commutation arguments gives pointwise boundedness and
decay statements. We give here an example
of the decay statements that follow:
\begin{corollary}
We have the following statement of pointwise decay for $\psi$
\begin{equation}\label{eq:pdecay}
|r^{\frac 12} \psi|\le C \sqrt{D}{\tau^{-1}},\qquad |r\psi|\le C \sqrt{D}{\tau^{-\frac 12}}.
\end{equation} 
 Moreover, for any $\eta>0$
\begin{equation}\label{eq:pdecay'}
|\psi|\le C_\eta \sqrt{D} \tau^{-\frac 32+\eta}.
\end{equation} 
\end{corollary}

Before closing this introduction, we give a few comments on some of the ideas
special to the Kerr case.

\subsection{The interaction of the red-shift, superradiance, and trapped null orbits}
We have already remarked in the Schwarzschild case the importance of understanding
the red-shift effect
and the role of trapped null geodesics for decay, and how these can in fact
be captured by energy methods, \emph{bypassing completely} mode analysis, in fact,
bypassing any reference to Fourier
decomposition and spherical harmonics. 

Considering now the Kerr case, 
there are two main additional difficulties: superradiance, discussed previously, and the fact that the structure of trapped null geodesics
is more complicated, at least when viewed purely in physical space.
These two problems are in some sense coupled, and this coupling can be viewed
as an additional, third difficulty. (See Section~\ref{duskolies}.)
It turns out that to understand these difficulties, 
it is indeed useful to reconsider
Carter's formal separation of $\psi$ into modes, but viewed from a slightly different 
perspective.

Let us note that the first
issue mentioned in Section~\ref{LS?} (associated to whether general solutions $\psi$
can be written as superpositions of modes) can be sidestepped
with the help of a cutoff function,
by
restricting consideration to a finite time slab.
In such a slab, one can indeed
represent a given finite-energy solution $\psi$ of $(\ref{eq:box})$
as a superposition of modes
$(\ref{modetype})$
localised at fixed frequency-triple $\omega\in\mathbb R$, $m$, $\ell$. 
If one is able to recover  a suitable \emph{quantitative} estimate on the solution
in this timeslab,
then by a bootstrap argument the estimate can easily be extended to be valid for all time.

Recall from our brief
discussion of the Schwarzschild case that the key to showing quantitative
bounds 
is constructing suitable energy currents
which capture the usual energy
conservation, the red-shift, and the obstruction of trapped null orbits. 
Key to the original boundedness result~\cite{dr6} on Kerr were two observations:
\begin{itemize}
\item
In the case $|a|\ll M$, superradiant frequencies are not trapped.
\item
When superradiance is controlled by a small parameter, then it can be
absorbed with the help of the redshift.
\end{itemize}
The two observations taken together allowed (see Section~\ref{boam}) an understanding of
boundedness without understanding trapping. In fact,
the boundedness results were obtained not just for exactly
Kerr metrics but for a wide class
of spacetimes whose metric is $C^1$-close to Schwarzschild.
The decay result~\cite{jnotes} (which is restricted to Kerr spacetimes) in
the $|a|\ll M$ case replaces the first observation above with 
a complete treatment of trapping. Here the full potential of the
microlocalisation provided by $(\ref{modetype})$ is used: In the range of frequencies relevant for trapping,
virial currents are chosen separately
for each frequency-triple $(\omega, m, \ell)$, so as to degenerate
at a unique value of $r$ depending on this triple. This $r$-value can be
related to a null orbit with conserved quantities which are in turn associated
to the frequency triple. Behind this
construction is 
the close connection between the separation of the
wave equation and the separation of the Hamilton-Jacobi equations.
(See Section~\ref{smadec}.)

The main new idea necessary for the general $|a|<M$ case
is to revisit also the first observation above. In the small $|a|\ll M$ case, the fact
that superradiant frequencies are not trapped 
follows essentially from the fact that
all future-trapped null geodesics eventually leave the ergoregion.
For general $|a|<M$ the latter is no longer the case. Nonetheless,
remarkably,
\begin{itemize}
\item
Superradiant frequencies are not trapped
for the entire $|a|<M$ range!
\end{itemize}
In the special (marginal) case of  frequency $\omega=0$, the above is related to the
fact that 
there are no \emph{trapped} null geodesics orthogonal to $\partial_t$. This latter fact plays
a role in recent work on black hole uniqueness~\cite{alexakis}.

The above observation allows us to adapt the second observation from before
to the case
where superradiance is not a small parameter and still couple the conserved
energy estimate, the redshift and 
understanding of trapping (the latter again via the microlocalisation
achieved by the classical separation) so as to obtain the desired estimates (see 
Section~\ref{edwgeviko}).

The resolution of the linear stability problem in the $|a|<M$ case
thus brings us full circle. Though based on energy methods,
the proof reconnects with the classical mode analysis,
now considered from a more sophisticated point of view as a tool for the `microlocalisation' of energy currents.
Moreover, the method highlights certain properties of
individual modes $(\ref{modetype})$, indeed of the potential functions $V(\omega)$,
which do not appear to have been 
noticed previously, but turn out
to play a fundamental part in the stability mechanism. Thus, we hope that this argument
offers something new even for the reader impatient with some of the more technical
difficulties of the proof.

\subsection{Outlook}
With Theorem~\ref{thrm:state} and its corollaries, 
the problem of linear stability, at least for scalar perturbations, has been essentially 
completely understood, leaving only the extremal $|a|=M$ case. In view of
recent work of Aretakis~\cite{aretakis} on extremal Reissner-Nordstr\"om, 
for extremal Kerr one in fact expects \emph{instabilities} on the horizon. 
We should mention that electromagnetic perturbations on Schwarzschild have been 
treated in \cite{Blue}, while the problem of gravitational 
perturbations is the subject of recent work of Holzegel~\cite{holzegel}.

\section{The Kerr space-time}
\label{KST}
The detailed description of the maximally
extended Kerr space-time $(\mathcal{M},g_{M,a})$
can be found for instance in \cite{cartersep},\cite{jnotes}. 
Here we simply recall the local coordinate expression of the metric
in the regions which will be of interest to us.

\subsection{Boyer-Lindquist coordinates}
The Kerr metric in the domain of outer communication
${\mathcal R}\subset \mathcal{M}$
(corresponding to one of the ends) is given in so-called \emph{Boyer-Lindquist coordinates} 
$(t,r,\theta,\phi)\in {\Bbb R}\times(r_+,\infty)\times [0,\pi]\times[0,2\pi)$
by
\begin{align}
\label{eleme}
g_{M,a}=
-\frac{\Delta}{\rho^2}\left(dt-a\sin^2\theta d\phi\right)^2
+\frac{\rho^2}{\Delta}dr^2+\rho^2d\theta^2 +\frac{\sin^2\theta}{\rho^2}
\left(a\,dt-(r^2+a^2)d\phi\right)^2.
\end{align}
Here 
$$
\Delta=r^2-2M r+a^2,\qquad \rho^2=r^2+a^2\cos^2\theta,
$$
and $r_+$ is the largest root of the equation $\Delta=0$,
\[
r_\pm=M\pm \sqrt{M^2-a^2};
\index{fixed parameters! $r$-parameters! $r_+$ (larger root of $\Delta=0$)}
\index{fixed parameters! $r$-parameters! $r_-$ (smaller root of $\Delta=0$)}
\]
The coordinate system restricted to $\theta\ne 0,\pi$, $\phi\ne 0$ defines an almost 
global chart on $\mathcal{R}$.
The coordinate vectorfields $T=\partial_t$ and $\Phi=\partial_\phi$ extend to
globally defined Killing fields 
and characterize Kerr as a \emph{stationary}, \emph{axi-symmetric} space-time.

Setting $a=0$, we recognize the  Schwarzschild metric
\begin{equation}
\label{incords}
g_M=-\left(1-\frac{2M}r\right) dt^2+\left(1-\frac{2M}r\right)^{-1}
dr^2 +r^2 (d\theta^2+\sin^2\theta \,d\phi^2),
\end{equation}
in standard form.

It is useful to define a rescaled version of the $r$ coordinate on $r>r_+$ by
\begin{equation}
\label{r*def}
\frac{dr^*}{dr}=\frac{r^2+a^2}{\Delta},
\end{equation}
after suitably chosing a value for $r^*=0$.
\index{coordinates! $r^*$ (Regge-Wheeler-type rescaled $r$ coordinate $r^*=r^*(r)$)}
We will sometimes consider the
$(t,r^*,\theta, \phi)$ coordinate system.
Note that these now range in ${\Bbb R}\times {\Bbb R}\times [0,\pi]\times[0,2\pi)$.

\subsection{Kerr-star coordinates}
For $|a|<M$, 
recall that the boundary $\partial\overline{\mathcal{R}}$ 
in the maximally extended Kerr spacetime $\mathcal{M}$
 is a bifurcate null hypersurface 
$\mathcal{H}=\mathcal{H}^+\cup\mathcal{H}^-$, where $\mathcal{H}^\pm$ are
known as the future (rsp.~past) event horizons.

The geometry of $\mathcal{H}$ is not easily understood in Boyer-Lindquist
coordinates, which degenerate as one approaches this hypersurface.
Defining, however, so-called Kerr star coordinates $(r^*, r,\theta,\phi^*)$
on $\mathcal{R}$ by
\[
t^*=t +\bar t(r),\qquad 
 \phi^*=\phi+\bar \phi(r)
\]
where 
\[
\frac{d \bar t}{dr}(r)= (r^2+a^2)/\Delta^2, \qquad 
\frac{d \bar\phi}{dr}(r) = a/\Delta,
\]
then these coordinates extend regularly to be defined on a larger 
subset of $\mathcal{M}$, covering in particular
$\mathcal{H}^+\setminus \mathcal{H}^-$, modulo the standard degeneration of the spherical coordinates on ${\Bbb S}^2$.
In these Kerr star coordinates, the set  ${\mathcal H}^+\setminus\mathcal{H}^-$ 
now simply corresponds to the set $\{r=r_+\}$.
We note that the Killing fields again take the form $T=\partial_{t^*}$, $\Phi=\partial_{\phi^*}$.

\section{Geometric phenomena of Schwarzschild and Kerr}
\label{phenomena}
In this section we briefly review the geometric features of Schwarzschild and
the Kerr space-time which will be particularly important for understanding
the behaviour of solutions the wave equation.

\subsection{The stationary Killing field and the ergoregion}
\label{ergosec}
We have already remarked that  for the Kerr family, the Boyer-Lindquist
coordinate vector field $T=\partial_t$ is Killing. We consider briefly its geometry.

We begin with the Schwarzschild case. From $(\ref{incords})$,
we immediately read off that $T=\partial_t$ is timelike everywhere 
in the domain of outer communications $\mathcal{R}$. On $\mathcal{H}^+\cup\mathcal{H}^-$,
$T$ becomes null, vanishing in fact on $\mathcal{H}^+\cap\mathcal{H}^-$. 
(In the black hole region, $T$ becomes
spacelike; this is of no concern to us here.)

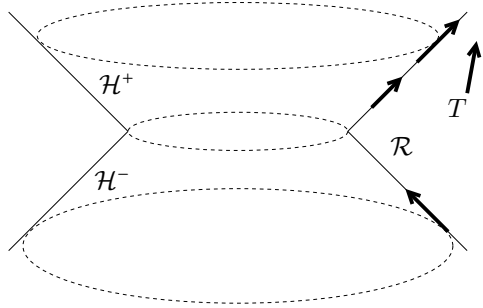
\[
\begin{picture}(0,0)%
\includegraphics{onhor.pstex}%
\end{picture}%
\setlength{\unitlength}{1579sp}%
\begingroup\makeatletter\ifx\SetFigFont\undefined%
\gdef\SetFigFont#1#2#3#4#5{%
  \reset@font\fontsize{#1}{#2pt}%
  \fontfamily{#3}\fontseries{#4}\fontshape{#5}%
  \selectfont}%
\fi\endgroup%
\begin{picture}(7511,4740)(2314,-6368)
\put(3751,-3061){\makebox(0,0)[lb]{\smash{{\SetFigFont{10}{12.0}{\rmdefault}{\mddefault}{\updefault}{\color[rgb]{0,0,0}$\mathcal{H}^+$}%
}}}}
\put(3676,-4561){\makebox(0,0)[lb]{\smash{{\SetFigFont{10}{12.0}{\rmdefault}{\mddefault}{\updefault}{\color[rgb]{0,0,0}$\mathcal{H}^-$}%
}}}}
\put(8326,-4036){\makebox(0,0)[lb]{\smash{{\SetFigFont{10}{12.0}{\rmdefault}{\mddefault}{\updefault}{\color[rgb]{0,0,0}$\mathcal{R}$}%
}}}}
\put(9226,-3436){\makebox(0,0)[lb]{\smash{{\SetFigFont{10}{12.0}{\rmdefault}{\mddefault}{\updefault}{\color[rgb]{0,0,0}$T$}%
}}}}
\end{picture}%

\]

As we shall see, the fact that $T$ becomes null is already a difficulty when
we discuss the notion of energy in the sections to come.
In the case of the Kerr metric $g_{M,a}$ for $a\ne 0$, the situation is far worse!
Now, the vectorfield $T$ fails to be everywhere timelike in the domain of outer 
communication $\mathcal{R}$.
The subset $\mathcal{E}\subset {\rm int}(\mathcal{R})$
where $T$ is spacelike is known
as the \emph{ergoregion}. The boundary $\partial\mathcal{E}$
of $\mathcal{E}$ (in the topology
of $\mathcal{R}$) is called the
\emph{ergosphere}. 

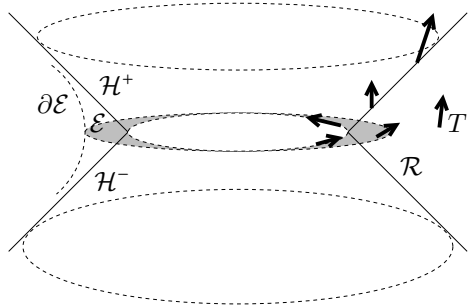
\[
\begin{picture}(0,0)%
\includegraphics{onhor2.pstex}%
\end{picture}%
\setlength{\unitlength}{1579sp}%
\begingroup\makeatletter\ifx\SetFigFont\undefined%
\gdef\SetFigFont#1#2#3#4#5{%
  \reset@font\fontsize{#1}{#2pt}%
  \fontfamily{#3}\fontseries{#4}\fontshape{#5}%
  \selectfont}%
\fi\endgroup%
\begin{picture}(7511,4740)(2314,-6368)
\put(3751,-3061){\makebox(0,0)[lb]{\smash{{\SetFigFont{10}{12.0}{\rmdefault}{\mddefault}{\updefault}{\color[rgb]{0,0,0}$\mathcal{H}^+$}%
}}}}
\put(3676,-4561){\makebox(0,0)[lb]{\smash{{\SetFigFont{10}{12.0}{\rmdefault}{\mddefault}{\updefault}{\color[rgb]{0,0,0}$\mathcal{H}^-$}%
}}}}
\put(9226,-3661){\makebox(0,0)[lb]{\smash{{\SetFigFont{10}{12.0}{\rmdefault}{\mddefault}{\updefault}{\color[rgb]{0,0,0}$T$}%
}}}}
\put(8476,-4336){\makebox(0,0)[lb]{\smash{{\SetFigFont{10}{12.0}{\rmdefault}{\mddefault}{\updefault}{\color[rgb]{0,0,0}$\mathcal{R}$}%
}}}}
\put(2776,-3361){\makebox(0,0)[lb]{\smash{{\SetFigFont{10}{12.0}{\rmdefault}{\mddefault}{\updefault}{\color[rgb]{0,0,0}$\partial\mathcal{E}$}%
}}}}
\put(3601,-3661){\makebox(0,0)[lb]{\smash{{\SetFigFont{10}{12.0}{\rmdefault}{\mddefault}{\updefault}{\color[rgb]{0,0,0}$\mathcal{E}$}%
}}}}
\end{picture}%

\]

Recall also the axisymmetric Killing field $\Phi$. We note that the combination
\[
L=T+  \frac{a}{r_+^2+a^2} \Phi
\]
is null and normal to $\mathcal{H}^+\setminus\mathcal{H}^-$.
Thus $\mathcal{H}^+\setminus\mathcal{H}^-$ 
is a Killing horizon.
The Killing field $L$ is sometimes known as the
\emph{Hawking vector field}.\index{vector fields! $L$ (Hawking vector field)}
We note the identity
\begin{equation}
\label{eqforK}
\nabla_L L=
\kappa L,
\end{equation}
where
\begin{equation}
\label{eutuxws}
\kappa= \frac{r_+-r_-}{2(r_+^2+a^2)}>0.
\end{equation}
The\index{fixed functions! spacetime
functions! $\kappa$ (surface gravity $\kappa=\kappa(a,M)$)} quantity $\kappa$ is known as the \emph{surface gravity}.\index{surface gravity}
We note that in the extremal case $|a|=M$ (not considered here!), $\kappa$ vanishes.

\subsection{Separability of geodesic flow and  trapped null geodesics}
\label{TNG}
The high frequency behavior of solutions to wave equations is intimately related to
the properties of null geodesic flow.
We thus consider these properties in the black hole case.

Again, we begin with the Schwarzschild case $a=0$.
In view of the dimensionality of the span of
the Killing fields $T, \Omega_1, \Omega_2, \Omega_3$,
it is immediate that  the Hamilton-Jacobi
equations separate, and thus, geodesic flow can be understood
completely (cf.~the discussion of separation of the wave equation in
Section~\ref{classical}). Geodesic flow in the Schwarzschild 
metric is described in detail in many textbooks.
One sees from the resulting equations that there are null geodesics $\gamma$ which
for all affine time remain on the hypersurface $r=3M$, the so-called 
\emph{photon sphere}.\index{photon sphere}
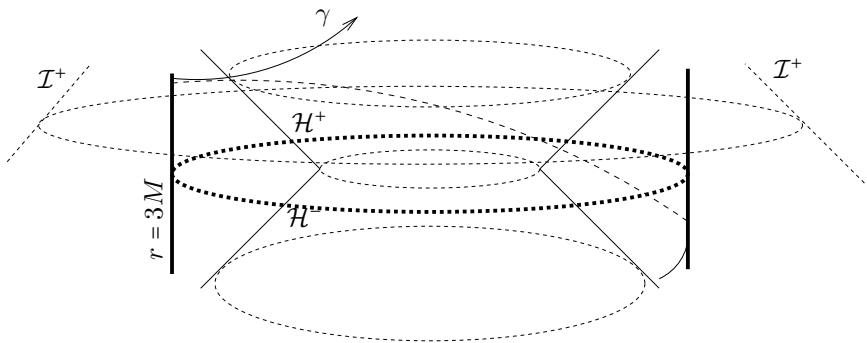
\[
\begin{picture}(0,0)%
\includegraphics{fortrap.pstex}%
\end{picture}%
\setlength{\unitlength}{1579sp}%
\begingroup\makeatletter\ifx\SetFigFont\undefined%
\gdef\SetFigFont#1#2#3#4#5{%
  \reset@font\fontsize{#1}{#2pt}%
  \fontfamily{#3}\fontseries{#4}\fontshape{#5}%
  \selectfont}%
\fi\endgroup%
\begin{picture}(13484,5308)(-724,-6368)
\put(3751,-3061){\makebox(0,0)[lb]{\smash{{\SetFigFont{10}{12.0}{\rmdefault}{\mddefault}{\updefault}{\color[rgb]{0,0,0}$\mathcal{H}^+$}%
}}}}
\put(3676,-4561){\makebox(0,0)[lb]{\smash{{\SetFigFont{10}{12.0}{\rmdefault}{\mddefault}{\updefault}{\color[rgb]{0,0,0}$\mathcal{H}^-$}%
}}}}
\put(-224,-2386){\makebox(0,0)[lb]{\smash{{\SetFigFont{10}{12.0}{\rmdefault}{\mddefault}{\updefault}{\color[rgb]{0,0,0}$\mathcal{I}^+$}%
}}}}
\put(11326,-2236){\makebox(0,0)[lb]{\smash{{\SetFigFont{10}{12.0}{\rmdefault}{\mddefault}{\updefault}{\color[rgb]{0,0,0}$\mathcal{I}^+$}%
}}}}
\put(1726,-5086){\rotatebox{90.0}{\makebox(0,0)[lb]{\smash{{\SetFigFont{10}{12.0}{\rmdefault}{\mddefault}{\updefault}{\color[rgb]{0,0,0}$r=3M$}%
}}}}}
\put(4126,-1336){\makebox(0,0)[lb]{\smash{{\SetFigFont{10}{12.0}{\rmdefault}{\mddefault}{\updefault}{\color[rgb]{0,0,0}$\gamma$}%
}}}}
\end{picture}%

\]
With reference to suitable asymptotic notions defining `future null infinity' $\mathcal{I}^+$, 
we can make the 
following more general statement:
If $\gamma(s)$ is a future-directed
inextendible null geodesic in Schwarzschild $(\mathcal{R},g_{M,0})$
with $\gamma(s)\in \mathcal{R}\setminus \mathcal{H}^+$
for all $s>s_0$,  such that moreover for all $p\in \mathcal{I}^+$, 
$\exists_s$ such that $\gamma(s)\not\in J^-(p)$, then $\lim_{s\to\infty} r(\gamma(s))=3M$.

Turning now to the general Kerr case,
remarkably, as discovered by Carter~\cite{cartersep},
geodesic flow admits, besides the conserved quantities associated
to the Killing fields $T$ and $\Phi$, a third non-trivial conserved quantity
(the `Carter constant'). 
Thus, 
geodesic flow remains 
separable and can be completely
understood.  Again, this separation is related to the separation
of the wave equation, discussed in Section~\ref{classical}.
The dynamics is described in some detail in~\cite{chandra}.

In contrast to the Schwarzschild case,
there are now null geodesics with constant $r$ for an open range of 
Boyer-Lindquist-$r$ values. However, restricting to geodesics sharing
a fixed triple of the nontrivial conserved quantities, 
there is at most one Boyer-Lindquist $r$-value, depending on the triple, to which all 
null geodesics neither crossing
$\mathcal{H}^+$, nor approaching $\mathcal{I}^+$, must necessarily asymptote
to towards the future.

\subsection{Red-shift effect}
\label{RSgeo}
An important stabilizing mechanism for the behaviour of waves near black hole event horizons
is what we shall here call the ``horizon-localized'' red-shift effect.  

Recall 
that this is the red-shift relating two observers $A$
and $B$, where $B=\varphi_\tau(A)$ for $\tau>0$, 
both crossing the event horizon.
In the geometric optics approximation, at horizon crossing time, the frequency of waves received by $B$ (measured with
respect to proper time) from
$A$  are damped (in comparison to the frequency measured by $A$) by a factor
exponentially decaying in $\tau$.
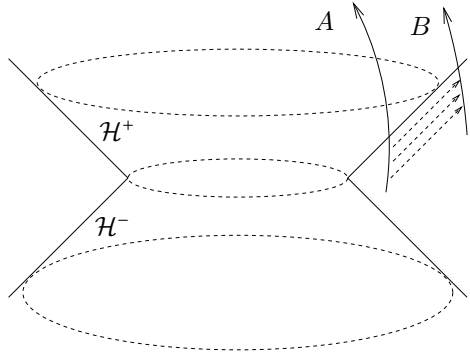
\[
\begin{picture}(0,0)%
\includegraphics{forred.pstex}%
\end{picture}%
\setlength{\unitlength}{1579sp}%
\begingroup\makeatletter\ifx\SetFigFont\undefined%
\gdef\SetFigFont#1#2#3#4#5{%
  \reset@font\fontsize{#1}{#2pt}%
  \fontfamily{#3}\fontseries{#4}\fontshape{#5}%
  \selectfont}%
\fi\endgroup%
\begin{picture}(7227,5497)(2314,-6368)
\put(3751,-3061){\makebox(0,0)[lb]{\smash{{\SetFigFont{10}{12.0}{\rmdefault}{\mddefault}{\updefault}{\color[rgb]{0,0,0}$\mathcal{H}^+$}%
}}}}
\put(3676,-4561){\makebox(0,0)[lb]{\smash{{\SetFigFont{10}{12.0}{\rmdefault}{\mddefault}{\updefault}{\color[rgb]{0,0,0}$\mathcal{H}^-$}%
}}}}
\put(7126,-1336){\makebox(0,0)[lb]{\smash{{\SetFigFont{10}{12.0}{\rmdefault}{\mddefault}{\updefault}{\color[rgb]{0,0,0}$A$}%
}}}}
\put(8626,-1411){\makebox(0,0)[lb]{\smash{{\SetFigFont{10}{12.0}{\rmdefault}{\mddefault}{\updefault}{\color[rgb]{0,0,0}$B$}%
}}}}
\end{picture}%

\]

The above effect depends in fact only on
the positivity of the \emph{surface gravity}, as was established for
non-extremal Kerr in Section~\ref{ergosec}. In the extremal case, we note
that the above horizon-crossing time red-shift degenerates.

\section{The energy method}
In this section we introduce the machinery of compatible currents. This is a robust and
invariant
approach to generating $L^2$-based identities, which can be exploited to establish links between 
the geometry of space-time and the behaviour of linear (and often even nonlinear)
waves.

\subsection{Compatible currents in physical space}
Given a Lorentzian metric $g$,
let $\Psi$ be sufficiently regular and satisfy
\[
\Box_g \Psi = F.
\]
We define
\[
{\bf T}_{\mu\nu}[\Psi]
\doteq \partial_\mu\Psi\partial_\nu\Psi -\frac12 g_{\mu\nu}g^{\alpha\beta}
\partial_\alpha\Psi \partial_\beta\Psi.
\]
\index{energy currents! $2$-currents! ${\bf T}_{\mu\nu}[\Psi]$ (energy-momentum tensor)}

Given a vector field $V^\mu$ and a function $w$,
we will define the currents
\[
{\bf J}^V_\mu[\Psi] = {\bf T}_{\mu\nu}[\Psi] V^\nu
\]
\index{energy currents! $1$-currents! ${\bf J}^V_\mu[\Psi]$ (energy $1$-current associated to vector field $V$)}
\[
{\bf J}^{V,w}_\mu[\Psi] = 
{\bf J}_\mu^V[\Psi]+\frac18w\partial_\mu (\Psi^2)-\frac18(\partial_\mu w)\Psi^2
\]
\index{energy currents! $1$-currents! ${\bf J}^{V,w}_\mu[\Psi]$ (modified energy $1$-current associated to vector field $V$,
function $w$)}
\[
{\bf K}^V[\Psi] ={\bf T}_{\mu\nu}[\Psi]\nabla^\mu V^\nu
\]
\index{energy currents! $0$-currents! ${\bf K}^V[\Psi]$ (energy $0$-current associated to vector field $V$)}
\[
{\bf K}^{V,w}[\Psi] = {\bf K}
^V[\Psi] -\frac18\Box_gw (\Psi^2) +\frac14w \nabla^\alpha\Psi\nabla_\alpha\Psi
\]
\index{energy currents! $0$-currents! ${\bf K}^{V,w}[\Psi]$ (modified energy $0$-current associated to vector field $V$,
function $w$)}
\[
\mathcal{E}^V[\Psi] = -FV^\nu \Psi_{,v}
\]
\[
\mathcal{E}^{V,w}[\Psi]= \mathcal{E}^V(\Psi)-\frac14w\Psi F
\]

Applying the divergence identity between two homologous spacelike
hypersurfaces $S_1$, $S_2$, bounding a region $\mathcal{B}$,
with $S_2$ in the future of $S_1$, we obtain 
\begin{equation}\label{eq:Not}
\int_{S_2}{\bf J}^V_\mu [\Psi]n^\mu_{S_2}+
\int_{\mathcal{B}} {\bf K}^V[\Psi] + \mathcal{E}^V[\Psi]
=\int_{S_1}{\bf J}^V_\mu [\Psi]n^\mu_{S_1},
\end{equation}
where $n_{S_i}$ denotes the future directed timelike unit normal.
A similar identity holds replacing $V$ by $V,w$.

\subsection{Noether's theorem}
The most celebrated use of $(\ref{eq:Not})$ is in the case where
$V$ is Killing, and $\Psi=\psi$ where $\psi$ satisfies $\Box_g\psi=0$.
One obtains then the conservation law
\begin{equation}
\label{Noethers}
\int_{S_2}{\bf J}^V_\mu [\psi]n^\mu_{S_2}
=\int_{S_1}{\bf J}^V_\mu [\psi]n^\mu_{S_1}.
\end{equation}
This is a version of \emph{Noether's theorem}.

\subsection{The dominant energy condition}
When $V$ is timelike, then one has moreover that ${\bf J}^V_\mu [\Psi]n_{S_i}^\mu$
is nonnegative definite, in fact 
\begin{equation}
\label{timelic}
{\bf J}^V_\mu [\Psi]n^\mu_{S}\ge c(S) \sum_{\a=0}^3 |\pa_\a\Psi|^2
\end{equation}
where $\pa_\a$ denote derivatives with respect to an arbitrary system of
local regular coordinates.
Thus, given a timelike Killing field, $(\ref{Noethers})$ and $(\ref{timelic})$
yield an a priori estimate for the $L^2$ norm on $S_2$ of all local derivatives of  solutions
$\Psi=\psi$ of the wave equation, in terms of a similar quantity evaluated on $S_1$.

One should not think, however, that this is the only use of identity $(\ref{eq:Not})$.
If $V$ is timelike, but not Killing, then one can estimate the bulk term 
\[
{\bf K}^V[\Psi] \le c(S) P{\bf J}^V_\mu[\Psi]n^\mu_S
\]
in terms of the boundary. For solutions of the wave equation, this 
allows one to prove that the left hand side of $(\ref{eq:Not})$, integrated over
$S_2(t)$
grows at most exponentially in an appropriate notion of time $t$. 
This is fundamental
in proving local existence type results.

\subsection{Virial identities}
Fundamental for the considerations of the present paper will be 
yet another way of viewing identity
$(\ref{eq:Not})$. While in Noether's theorem and in the proof
of local existence, one often considers the  bulk term ${\bf K}^V[\Psi]$
as an error term, $(\ref{eq:Not})$ is also very useful when one can show positivity
properties for the bulk term ${\bf K}^V[\Psi]$. Then, provided that
the boundary terms are also controlled, $(\ref{eq:Not})$ can be used 
to show the boundedness of a spacetime integral quantity. This can be thought
of as a weak statement of decay.
This is reminiscent of the use of the virial theorem.

\subsection{The hypersurfaces $\Sigma_\tau$}\label{sec:sigma}
Returning now to the Kerr geometry, we
 shall typically apply $(\ref{eq:Not})$ in regions bound between the event horizon,
 null infinity, and the leaves
 of a foliation associated to
hypersurfaces $\Sigma_\tau$ to be defined below.

Let $r_+<r_n<R_n<\infty$. The hypersurface $\Sigma_0$ is constructed as follows. Define the function
$f=f(r^*)$ according to 
$$
f(r^*)=\int_0^{r^*} \sqrt{1-\frac {M^2\Delta}{(r^2+a^2)^2}} 
$$
Then
\begin{align*}
&\Sigma_0=\{t+f(r^*)=0\},\qquad\qquad\qquad\qquad\qquad\hskip .5pc{\text{for}}\,\,\,r_+<r<r_n,\\
&\Sigma_0=\{t=-f(r^*_n)\},\qquad\qquad\qquad\qquad\hskip .5pc {\text{for}}\,\,\,r_n<r<R_n,\\
&\Sigma_0=\{t=f(r^*)-(f(r^*_n)+f(R^*_n))\},\qquad {\text{for}}\,\,\,R_n<r.
\end{align*} 

Now let $\varphi_\tau$ denote the $1$-parameter family of
diffeomorphisms generated by the Killing field  $T$  and define
the hypersurfaces $\Sigma_\tau:=\varphi_\tau(\Sigma_0)$.
(Equivalently, we replace everywhere $t$ by $t-\tau$ in the definition of $\Sigma_0$ above.)
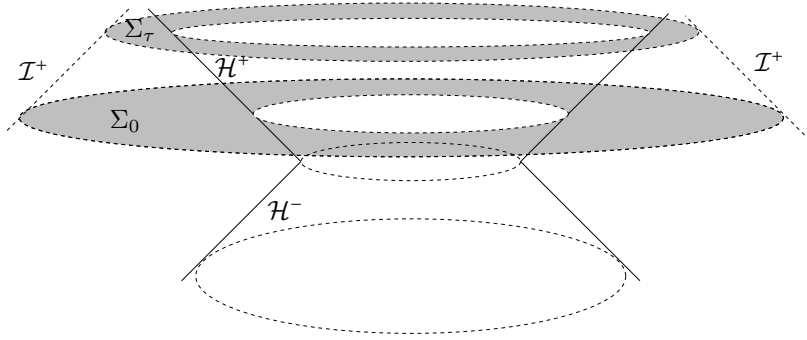
\[
\begin{picture}(0,0)%
\includegraphics{forsig.pstex}%
\end{picture}%
\setlength{\unitlength}{1579sp}%
\begingroup\makeatletter\ifx\SetFigFont\undefined%
\gdef\SetFigFont#1#2#3#4#5{%
  \reset@font\fontsize{#1}{#2pt}%
  \fontfamily{#3}\fontseries{#4}\fontshape{#5}%
  \selectfont}%
\fi\endgroup%
\begin{picture}(12546,5204)(-424,-6368)
\put(3676,-4561){\makebox(0,0)[lb]{\smash{{\SetFigFont{10}{12.0}{\rmdefault}{\mddefault}{\updefault}{\color[rgb]{0,0,0}$\mathcal{H}^-$}%
}}}}
\put(-224,-2386){\makebox(0,0)[lb]{\smash{{\SetFigFont{10}{12.0}{\rmdefault}{\mddefault}{\updefault}{\color[rgb]{0,0,0}$\mathcal{I}^+$}%
}}}}
\put(11326,-2236){\makebox(0,0)[lb]{\smash{{\SetFigFont{10}{12.0}{\rmdefault}{\mddefault}{\updefault}{\color[rgb]{0,0,0}$\mathcal{I}^+$}%
}}}}
\put(1201,-3136){\makebox(0,0)[lb]{\smash{{\SetFigFont{10}{12.0}{\rmdefault}{\mddefault}{\updefault}{\color[rgb]{0,0,0}$\Sigma_0$}%
}}}}
\put(2851,-2311){\makebox(0,0)[lb]{\smash{{\SetFigFont{10}{12.0}{\rmdefault}{\mddefault}{\updefault}{\color[rgb]{0,0,0}$\mathcal{H}^+$}%
}}}}
\put(1426,-1711){\makebox(0,0)[lb]{\smash{{\SetFigFont{10}{12.0}{\rmdefault}{\mddefault}{\updefault}{\color[rgb]{0,0,0}$\Sigma_\tau$}%
}}}}
\end{picture}%

\]
We may now denote by $\tau$ the function on $J^+(\Sigma_0)$ which is
constant on the level sets $\Sigma_\tau$.
The hypersurfaces $\Sigma_\tau$ are space-like: 
$$
g(\nab\tau,\nab\tau)\le -\frac {M^2-a^2}{2\rho^2}.
$$
Near null infinity 
$$
\tau=t-r^*+ O(M^2 r^{-2}),\qquad r\to \infty
$$
while near the horizon
$$
\tau=t^*+ O(M^2),\qquad r\to r_+.
$$

By the coarea formula, we easily see that, denoting by
$\mathcal{B}$ the region bounded by $\Sigma_{\tau_1}$ and $\Sigma_{\tau_2}$
and the $\mathcal{H}^+_{(\tau_1,\tau_2)}$, $\mathcal{I}^+_{(\tau_1,\tau_2)}$,
we have that for an arbitrary function $F$,
\[
\int_{\mathcal{B}} F\sim \int_{\tau_1}^{\tau_2}\left(\int_{\Sigma_\tau} F \right) d\tau
\]
When the measure of integration is not denoted, the volume form is always meant.
We will use this relation  without comment in what follows.

\subsection{Non-degenerate derivatives}\label{sec:nder}
Finally, since energy estimates will control various combinations of derivatives,
it is useful to  have suggestive notations for these.

First, we shall denote $\nabb$ the induced gradient on the $r=$constant spheres.

We shall use the notation $|\pa\psi|^2$ to denote the type of
quantity controlled by the flux of ${\bf J}^{N}$ through $\Sigma_\tau$ for $\tau\ge 0$,
where $N$ is a $\varphi_t$-invariant vector field
which is strictly timelike and $N=T$ for large $r$ (for instance, we may take the
vector field of Proposition~\ref{specialises..}).
Explicitly, we may define
\[
|\pa \psi|^2=(\partial_r\psi)^2 +(\partial_{t^*}\psi)^2 +|\nabb\psi|^2
\]
for $r\le r_n$, where the coordinate derivatives are interpreted in Kerr-star coordinates
and
\[
|\pa\psi|^2=\alpha(r)(\partial_{t^*}\psi)^2+(\partial_r\psi)^2+|\nabb\psi|^2
\]
for $r\ge r_n$, where $\alpha(r)$ is a positive function which vanishes as $r\to \infty$.

Finally, let us use the notation $|\overline\pa\psi|^2$ to denote the derivatives
appearing in the flux of an $N$ as above on the null hypersurfaces $\mathcal{H}^+$ 
and $\mathcal{I}^+$, respectively. These are all tangential derivatives.
We may write explicitly
\[
|\overline\pa\psi|^2= |L\psi|^2 +|\nabb\psi|^2
\]
on $\mathcal{H}^+$
and 
\[
|\overline\pa\psi|^2 = |T\psi|^2+|\nabb\psi|^2.
\]
on $\mathcal{I}^+$.

\section{Schwarzschild: the three sources and components of decay}\label{sec:three}
In this section, we describe three ingredients of the recipe for the robust derivation 
of the quantitative decay rates for solutions of the wave equation on Schwarzschild black hole
space-times. These are directly related to the geometric phenomena studied in
Section~\ref{phenomena}.
We will see in the following section how these
ingredients are used to show `integrated local energy decay'
and, together with an extra ingredient of a much more general type, 
actual decay.

The Kerr problem to be addressed later will require revisiting the strategy and, among
other things, microlocalising the 
components of this section in a manner strongly attached to the Kerr geometry. 
See Section~\ref{revisit}.

\subsection{The conserved energy estimate}
\label{Conservative}
The Schwarzschild metric is stationary so we
may apply Noether's theorem $(\ref{Noethers})$
for $V=T$, where $T$ is the `stationary' Killing field.
Specifically, we apply $(\ref{Noethers})$ to
the space-time region bounded by two space-like hypersurfaces $\Sigma_{\tau_1}$, $\Sigma_{\tau_2}$,
part of the event horizon ${\mathcal H}^+_{(\tau_1,\tau_2)}$ and null infinity ${\mathcal I}^+_{(\tau_1,\tau_2)}$,
i.e.~we take $S_2={\mathcal H}^+_{(\tau_1,\tau_2)}\cup \Sigma_{\tau_2}\cup
{\mathcal I}^+_{(\tau_1,\tau_2)}$, $S_1=\Sigma_{\tau_1}$. We obtain:
\begin{equation}
\label{Nomos}
\int_{\Sigma_{\tau_2}}{\bf J}^T_\mu [\Psi]n^\mu_{\Sigma_\tau}+ 
\int_{{\mathcal H}^+_{(\tau_1,\tau_2)}} {\bf J}^T_\mu [\Psi]n^\mu_{{\mathcal H}^+}
+ \int_{{\mathcal I}^+_{(\tau_1,\tau_2)}} {\bf J}^T_\mu [\Psi]n^\mu_{{\mathcal I}^+}
=\int_{\Sigma_{\tau_1}}{\bf J}^T_\mu [\Psi]n^\mu_{\Sigma_\tau}.
\end{equation}

We shall see in Section~\ref{supersec} that for the general Kerr case, the above
identity is a priori of limited use in view of the failure of $T$ to be timelike
in the exterior. 
Already, however,  in the case of the Schwarzschild solution with $a=0$, 
the vectorfield $T$,
though timelike in the domain of 
outer communication $\mathcal{R}$, becomes null on the horizon 
$\mathcal{H}^+=\{r=2M\}$. Thus, while all integrands in $(\ref{Nomos})$ are nonnegative
and thus each integral taken individually is controlled by the right hand side,  
the first two integrands degenerate with respect to what one would have were $T$ timelike.
For the first integrand, we thus only have the degenerate lower bound
\begin{equation}
\label{3eftila1}
{\bf J}^T_\mu [\Psi]n^\mu_{\Sigma_\tau}(r)\ge c(r)|\pa\Psi|^2,\qquad
\end{equation}
where $c(r)=0$ for $r=2M$,
and similarly, for the second integrand, we only have
\begin{equation}
\label{3eftila2}
{\bf J}^T_\mu [\Psi]n^\mu_{{\mathcal H}^+}\ge c |T\Psi|^2,
\end{equation}
i.e.~the right hand side is missing\footnote{cf.~the `expected' $|\overline\pa\Psi|^2$}
 the term
$|\nabb\Psi|^2$.

From $(\ref{Nomos})$, we thus obtain the estimates
\begin{equation}
\label{onspacelike}
\int_{\Sigma_\tau} c(r)|\partial\psi|^2 \le \int_{\Sigma_0} {\bf J}^T_\mu[\Psi]n^\mu,
\end{equation}
\begin{equation}
\label{onhori}
\int_{\mathcal{H}^+}|T\psi|^2 \le \int_{\Sigma_0} {\bf J}^T_\mu[\Psi]n^\mu,
\end{equation}
\begin{equation}
\label{onnoufou}
\int_{\mathcal{I}^+}|\overline\partial \psi|^2 \le \int_{\Sigma_0} {\bf J}^T_\mu[\Psi]n^\mu.
\end{equation}

 The standard approach to proving pointwise boundedness of solutions $\psi$
 of $(\ref{eq:box})$ is by proving
 $L^2$-estimates for \emph{all} derivatives of $\psi$ up to sufficiently high order. 
 Thus, even in the Schwarzschild case, the above degeneration
 of $(\ref{onspacelike})$ or $(\ref{onhori})$ presents significant complications
 for proving stability \emph{at the horizon}.
See Section~\ref{aside}.

\subsection{Virial estimates and trapping}
\label{virialss}
For the  discussion to follow here, it is useful to compare with Minkowski space.
Defining a current ${\bf J}^{X,w}[\Psi]$ where 
\begin{equation}
\label{oftheformc}
X= f(r)\partial_r
\end{equation}
for specific choices of $f(r)$ and $w(r)$, 
one has that for suitable regions $\mathcal{B}$
\begin{equation}
\label{lowerbnd0}
\int_{\mathcal{B}}{\bf K}^{X,w}[\Psi] \ge C_R\int_{\mathcal{B}\cap\{r\le R\}} 
(|\partial \Psi|^2+ |\Psi|^2),
\end{equation}
and thus, for solutions $\psi$ of the wave equation, identity $(\ref{eq:Not})$ yields
\begin{equation}
\label{yields}
\int_{\mathcal{B}\cap\{r\le R\}} (|\partial \psi|^2+ |\psi|^2)  \le   C \left| \int_{S_1}{\bf J}^{X,w}_\mu [\psi] n^\mu\right| +
C\left|\int_{S_2}{\bf J}^{X,w}_\mu [\psi] n^\mu \right|.
\end{equation}

The above relation together with energy conservation (see Section~\ref{ILED}) will already
yield a weak statement of decay. This for instance already excludes
stationary solutions. 

It turns out that for more general metrics, trapped null geodesics are a ``high frequency''
obstruction
for the existence of currents 
satisfying $(\ref{lowerbnd0})$.  In particular, in the Schwarzschild case, in view
of the considerations of Section~\ref{TNG}, one can show that there {\bf \emph{cannot}} 
exist a current ${\bf J}^{X,w}$
associated to a vector field $X$ of the form $(\ref{oftheformc})$
satisfying $(\ref{lowerbnd0})$ near $r=3M$. Cf.~classical results due to Ralston~\cite{ralston}.

Remarkably, it turns out that this trapping obstruction is the only obstruction
for the validity of $(\ref{lowerbnd0})$ away from the horizon, 
and one can construct a current
${\bf J}^{X,w}$ where $X$ is of the form $(\ref{oftheformc})$ satisfying 
\begin{equation}
\label{lowerbnd}
\int_{\mathcal{B}} {\bf K}^{X,w}[\Psi] \ge C_{R_1,R_2} \int_{\mathcal{B}\cap\{R_1\le r\le R_2\}  }((r-3M)^2 |\partial \Psi|^2+(\partial_r\Psi)^2+|\Psi|^2),
\end{equation}
for all $2M<R_1< R_2<\infty$.
That is to say, apart from the horizon and infinity, the control degenerates precisely on
the \emph{photon sphere} $r=3M$.

The vector field $X$ itself vanishes at $r=3M$. The construction of the
energy current ${\bf J}^{X,w}$ is quite delicate, and the original constructions in fact
required spherical harmonic decompositions, a requirement that was later overcome.
See~\cite{bluester, dr3, dr5, bluesof} and subsequent papers.

Applying now $(\ref{eq:Not})$  in the region
described in the previous section to solution $\psi$ of the wave equation $(\ref{eq:box})$,
one obtains
\begin{align}
\nonumber
\label{viresti}
 \int_{\tau_1}^{\tau_2}\left( \int_{\Sigma_\tau\cap \{R_1\le r\le R_2\}  }((r-3M)^2 |\partial \psi|^2+(\partial_r\psi)^2+|\psi|^2)\right)d\tau\\
\le C_{R_1,R_2} \left(\int_{\Sigma_{\tau_1}}+\int_{\Sigma_{\tau_2}}+\int_{\mathcal{H}^+_{(\tau_1,\tau_2)}}
+\int_{\mathcal{I}^+_{(\tau_1,\tau_2)}}\right)
|{\bf J}^{X,w}_\mu[\psi]n^\mu|.
\end{align}

\subsection{The red-shift estimate}
\label{R-S-E}
To properly understand the behaviour of waves on the horizon, we have to take account
of the red-shift effect described in Section~\ref{RSgeo}.
In the present section, we define a timelike vector field $N$\index{vector fields! 
$N$ (`red-shift' vector field)} 
whose multiplier 
current ${\bf J}^N$ captures precisely this effect.
This type of current was first discovered in the Schwarzschild case in~\cite{dr3}, but it turns
out, as was shown in~\cite{jnotes}, that the existence of this current is in fact a general property of stationary
black hole spacetimes
with Killing horizons of positive surface gravity.

In anticipation of later sections, we might as well at the onset give a statement
applicable to the
entire Kerr family, as the Schwarzschild case is no simpler.
Theorem~7.1 of~\cite{jnotes}, specialised then to Kerr, yields
\begin{proposition}
\label{specialises..}
Let $|a|<M$, $g_{a, M}$ be the Kerr metric and $\mathcal{R}$, etc., be as before.
There exist positive constants $b=b(a,M)$ and $B=B(a,M)$, 
parameters $r_1(a,M)>r_0(a,M)>r_+$, and
\index{fixed parameters! $r$-parameters! $r_0$ (associated to $N$)}
\index{fixed parameters! $r$-parameters! $r_1$ (associated to $N$)}
a $\varphi_t$-invariant timelike vector field $N=N(a,M)$ 
on $\mathcal{R}$ such that
\begin{enumerate}
\item
\label{fir}
${\bf K}^N[\Psi] \ge b\, {\bf J}^N_\mu [\Psi]  N^\mu$ for $r\le r_0$ 
\item 
\label{item2}
$-{\bf K}^N[\Psi] \le B\, {\bf J}^N_\mu [\Psi] N^\mu$, for $r\ge r_0$
\item
\label{lastitem}
$N= T$ for $r\ge r_1$,
\end{enumerate}
where the currents are defined with respect to $g_{M,a}$.
\end{proposition} 

Property~\ref{fir} above can be thought to embody the exponential damping property
discussed in Section~\ref{RSgeo}.

Applying $(\ref{eq:Not})$ with current ${\bf J}^N$ in the region of Section~\ref{Conservative},
we obtain using the above `red-shift' Proposition~\ref{specialises..} the following estimate for solutions
of the homogeneous wave equation $\Box_g \psi=0$ on Kerr:
\begin{align}\label{eq:rshift}
  \int_{\Sigma_{\tau_2}} |\pa\psi|^2&+
\int_{{\mathcal H}^+_{(\tau_1,\tau_2)}} |\overline\pa\psi|^2+\int_{{\mathcal I}^+_{(\tau_1,\tau_2)}} |\overline\pa\psi|^2\\ &+
\int_{\tau_1}^{\tau_2} \int_{\Sigma_\tau\cap\{r\le r++s_-\}} |\pa\psi|^2\le 
C\left ( \int_{\Sigma_{\tau_1}} |\pa\psi|^2+\int_{\tau_1}^{\tau_2} \int_{\Sigma_\tau\cap\{r_++s_-\le r\le r++s_+\}} |\pa\psi|^2\right).
\notag
\end{align}
Here $r_++s_-<r_0$ and $r_++s_+>r_1$. The red-shift effect is responsible for the horizon and the space-time term on the 
left hand-side of \eqref{eq:rshift}. It is generally effective only in a small neighborhood of the horizon. The term at null infinity
appears as a consequence of $N=T$ being time-like near ${\mathcal I}^+$.

We note that the above estimate degenerates in the extremal case $|a|=M$,
in accordance with the comments at the end of Section~\ref{RSgeo}. This
is in fact the fundamental new
 difficulty of the extremal problem. See~\cite{aretakis}.

The other manifestation of the red-shift effect is the following commutation statement
\begin{proposition}
\label{rsc}
Let $g=g_{M,a}$ and $N$ be as in Proposition~\ref{specialises..}.  Then,
on $\mathcal{H}^+$ with a distinguished generator $L$, extending $L$
to a translation invariant standard null frame $E_1,E_2, L, Y$, then
for all $k\ge 0$
\[
\Box_g(Y^k\Psi)=\kappa_k Y^{k+1}\Psi + \sum_{|{\bf m}|\le k+1, m_4\le k} c_{{\bf m}} 
E_1^{m_1}E_2^{m_2}L^{m_3}Y^{m_4}\Psi
\]
where $\kappa_k>0$.
\end{proposition}

In the Schwarzschild case, $Y$ and $T$ together will provide around every point
a timelike commutator with good properties. This is important for proving pointwise
bounds. These commutations will  also suffice in the $|a|\ll M$ case. In the general Kerr case, 
however, besides $Y$ and $T$, we shall also need to
commute with the axisymmetric Killing field $\Phi$ to obtain pointwise bounds, exploiting
the fact that the span of $T$ and $\Phi$ is timelike
except at the horizon, where the span of $T, \Phi, Y$ again contains a timelike direction.

\section{Schwarzschild: Putting it all together}\label{putting}
We now proceed to show how the above components can be used to show
decay type statements, beginning with integrated local energy decay.

\subsection{Integrated local energy decay}
\label{ILED}
The choice of the virial current ${\bf J}^{X,w}$ in Section~\ref{virialss}
is such that the boundary
terms 
\begin{equation}
\label{whattodo}
\left(\int_{\Sigma_{\tau_1}}+\int_{\Sigma_{\tau_2}}+\int_{\mathcal{H}^+_{(\tau_1,\tau_2)}}
+\int_{\mathcal{I}^+_{(\tau_1,\tau_2)}}\right)
|{\bf J}^{X,w}_\mu[\psi]n^\mu|
\end{equation}
can be controlled (with the help also of a Hardy inequality)
by
\[ 
\left( \int_{\Sigma_{\tau_1}}+\int_{\Sigma_{\tau_2}}+\int_{\mathcal{H}^+_{(\tau_1,\tau_2)}}
+\int_{\mathcal{I}^+_{(\tau_1,\tau_2)}}\right)
{\bf J}^{T}_\mu[\psi]n^\mu.
\]
Putting together $(\ref{viresti})$ and the conservation of energy
expressed by $(\ref{Nomos})$ we thus obtain
\begin{align}
\label{finaledw}
\nonumber
 \int_{\tau_1}^{\tau_2}&\left( \int_{\Sigma_\tau\cap \{R_1\le r\le R_2\}  }((r-3M)^2 |\partial \psi|^2+(\partial_r\psi)^2+|\psi|^2)\right)d\tau\\
&\le C_{R_1,R_2}\int_{\Sigma_{\tau_1}} {\bf J}^T_\mu [\psi]n^\mu.
\end{align}

From this statement, one can retrieve the usual integrated local decay (away from
the horizon) with
a ``loss of a derivative'', namely:
\begin{equation}
\label{apwleia}
\int_{\tau_1}^{\tau_2} \int_{\Sigma_t\cap \{R_1\le r\le R_2\}}
\left (|\pa\psi|^2+|\psi|^2\right)\le C_{R_1,R_2} \int_{\Sigma_{t_1}}( {\bf J}^T_\mu[\psi] 
+{\bf J}^T_\mu[T\psi]) n^\mu_{\Sigma_t}
\end{equation}
or alternatively
\begin{equation}
\label{oralt}
\int_{\tau_1}^{\tau_2} \int_{\Sigma_t\cap \{R_1\le r\le R_2\}}
\left (|\pa\psi|^2+|\psi|^2\right)\le C_{R_1,R_2} \int_{\Sigma_{t_1}}( {\bf J}^T_\mu[\psi] 
+\sum_{i=1}^3{\bf J}^T_\mu[\Omega_i\psi]) n^\mu_{\Sigma_t}.
\end{equation}
It is perhaps worth noting that if $P_\ell\psi$ denotes the projection of $\psi$ to the
space of spherical harmonics with number $\ell$, 
the second term on the right
hand side of $(\ref{oralt})$ can be dominated by the first. Thus, we obtain 
precisely $(\ref{finaledw})$
without the degeneration at $r=3M$. Fixed spherical harmonics thus do not see the trapping
obstruction.

Finally,  estimate $(\ref{finaledw})$ together with the red-shift estimate $(\ref{eq:rshift})$
yields the non-degenerate (cf.~the degenerate $(\ref{onspacelike})$) 
energy boundedness statement
\begin{equation}
\label{actuallybond}
\int_{\Sigma_{\tau_2}}|\pa\psi|^2\le C \int_{\Sigma_{\tau_1}}|\pa\psi|^2
\end{equation}
and
\begin{align}
\label{TRUE0}
\nonumber
 \int_{\tau_1}^{\tau_2}&\left( \int_{\Sigma_\tau\cap \{2M\le r\le R_2\}  }((r-3M)^2 |\partial \Psi|^2+(\partial_r\Psi)^2+|\Psi|^2)\right)d\tau\\
&\le C_{R_2}\int_{\Sigma_{\tau_1}}|\pa\psi|^2,
\end{align}
\begin{align}
\label{TRUE}
\nonumber
 \int_{\tau_1}^{\tau_2}&\left( \int_{\Sigma_\tau\cap \{2M\le r\le R_2\}  }(|\partial \psi|^2+
 (\partial_r\psi)^2+|\psi|^2)\right)d\tau\\
&\le C_{R_2}\int_{\Sigma_{\tau_1}}(|\pa(T\psi)|^2+|\pa\psi|^2).
\end{align}
The latter statement is a true integrated local energy decay,
without unnecessary degeneration at the event horizon.

\subsection{Aside: boundedness without the virial identity}
\label{aside}
It is a nice exercise for the reader to see that $(\ref{actuallybond})$ can be derived
using only the red-shift estimate $(\ref{eq:rshift})$ and conservation of energy $(\ref{Nomos})$, i.e.~without appeal to
$(\ref{finaledw})$.

Together with the red-shift commutation Proposition~\ref{rsc}, 
estimate $(\ref{actuallybond})$
allows one to prove uniform pointwise boundedness for $\psi$ and all its derivatives
to arbitrary order. See~\cite{dr6}. This in particular already
retrieves (and strengthens) the celebrated
Kay--Wald Theorem stated in the introduction.

\subsection{A final ingredient: estimates at null infinity}
The final ingredient for the statement of decay is a family of estimates 
derived in the neighborhood of future 
null infinity ${\mathcal I}^+$.  These estimates have been 
introduced in \cite{icmp} and in a way serve as a replacement of the 
older and less flexible approach involving the current ${\bf J}^Z$ 
with a vectorfield $Z$ constructed as an analog
of the conformal vectorfield $Z=(t^2+r^2) \frac {\pa}{\pa t} +2tr\frac {\pa}{\pa r}$ of Minkowski space.

The family of estimates is based on the application of the current ${\bf J}^V$ with a vectorfield 
$V=r^p \frac {\pa}{\pa v}$ composed from a power of the Schwarzschild $r$-coordinate and the 
null coordinate vector field $\frac {\pa}{\pa v}=\frac{\pa}{\pa t}+\frac {\pa}{\pa r^*}$.
Parameter $p$ is allowed to vary in the range $0\le p\le 2$. 
The estimates are cleaner when $\Sigma_\tau$ is in fact exactly null in $r\ge R_3$ for some
sufficiently large $R_3$.
The resulting estimate
for solutions $\psi$ of the wave equation then has the form
\begin{align}\label{eq:newmeth}
\int_{\Sigma_{\tau_2}\cap \{r\ge R_3\}} &r^{p-2} |\pa_v(r\psi)|^2 + \int_{\tau_1}^{\tau_2}
\int_{\Sigma_{\tau}\cap \{r\ge R_3\}} \left(p\,r^{p-3} |\pa_v(r\psi)|^2+(2-p)r^{p-1} |\nabb \psi|^2\right)\\ &\le
C \left (\int_{\Sigma_{\tau_1}\cap \{r\ge R_3\}} r^{p-2} |\pa_v(r\psi)|^2+ \int_{\tau_1}^{\tau_2}
\int_{\Sigma_{\tau}\cap \{R_2\le r\le R_3\}}
|\pa\psi|^2+|\psi|^2 \right),\notag
\end{align} 
and holds with $R_2<R_3$.

\subsection{Decay of the energy flux and pointwise decay}
We begin with the case of Minkowski space, where
already,  the $p$-hierarchy
of estimates $(\ref{eq:newmeth})$, coupled with the statement of 
integrated local energy decay $(\ref{yields})$ and the energy conservation associated
$\partial_t$, provides a new approach 
to the problem of robust decay.

Applying $(\ref{eq:newmeth})$ first for $p=2$, and using the integrated
decay estimate to control the second term on the right hand side, one may
extract (by the pigeonhole principle) a suitable sequence of hypersurfaces
$\Sigma_{\tau_i}$  for which
the first term on the right hand side corresponding to the $p=1$
estimate is bounded with $\tau_i$ in the role of $\tau_1$.
Repeating this argument and using also conservation of the usual energy,
one obtains easily:
\begin{equation}\label{eq:endecay}
\int_{\Sigma_\tau} |\pa\psi|^2\le \frac C{\tau^2} \int_{\Sigma_0} r^2 |\pa\psi|^2.
\end{equation} 
 This leads, after appropriate commutation arguments, to the pointwise decay 
\begin{equation}\label{eq:pdecayMini0}
|r^{\frac 12} \psi|\le \frac C{\tau},\qquad |r\psi|\le \frac {C}{\tau^{\frac 12}}
\end{equation} 
with a constant $C$ depending on the initial data. See~\cite{icmp}.

A further refinement would also give  
\begin{equation}\label{eq:pdecayMini}
|\psi|\le \frac C{\tau^{\frac 32}}.
\end{equation} 
We note that the estimate~\eqref{eq:pdecayMini} 
matches the strongest rate of decay 
obtainable by the full strength of the usual vector field method.
Here, however, we have avoided the use of either the conformal multiplier
or commuting vector fields generating 
Lorentz rotations and scaling, i.e.~multiplier currents or commutators with weights in $t$.
This is important
for the black hole case, in particular Kerr with $|a|<M$.\footnote{Cf.~the loss
in the decay rate proportional to $a$ in the original decay result for $|a|\ll M$
in~\cite{jnotes}.} Moreover, in view of this additional robustness,
this new method thus has an even a wider range of applications to nonlinear problems.
See~\cite{Yang}.

In the Schwarzschild case, the argument requires in addition commutation with $T$
in view of the loss of $T$ on the right hand side of $(\ref{TRUE})$. Other than that,
it is identical to the Minkowski case!
Combining~\eqref{actuallybond},~\eqref{TRUE} and \eqref{eq:newmeth} one gets, through an analogous iteration process to the above,
the decay of non-degenerate energy
\begin{equation}\label{eq:endecayScwh}
\int_{\Sigma_\tau} |\pa\psi|^2\le \frac C{\tau^2} \int_{\Sigma_0}( r^2 |\pa TT\psi|^2
+r^2|\pa T\psi|^2+r^2|\pa \psi|^2).
\end{equation}

One then obtains $(\ref{eq:pdecayMini0})$ by standard commutations. Again, one
can then obtain further refinements like $(\ref{eq:pdecayMini})$ without introducing
$t$-weights.

\section{Kerr: The strategy revisited}
\label{revisit}
In the Kerr case, we must revisit the Schwarzschild strategy in view
of the new additional difficulties arising from its more complicated geometry.

\subsection{The difficulties}\label{duskolies}

\subsubsection{Superradiance}
\label{supersec}
Recall the discussion of the ergoregion $\mathcal{E}$ in Section~\ref{ergosec}.
For $p\in\mathcal{E}$, the energy density
 ${\bf J}^T_\mu [\Psi] n^\mu(p)$ in general fails to be nonnegative. 
 Since if $a\ne0$, $\mathcal{E}\cap\Sigma_\tau\ne\emptyset$
 and $\mathcal{E}\cap\mathcal{H}^+\ne\emptyset$, it follows that 
 the first two integrals on the left hand side of $(\ref{Nomos})$ no longer
have positive definite integrands, and 
 $(\ref{Nomos})$ no longer provides a priori control on $\Psi$.
 
In particular, the 
 flux of energy to null infinity--the third term on the left hand side of $(\ref{Nomos})$,
 which \emph{is} positive definite--can thus be larger than the initial energy,
 i.e.~the term on the right hand side.
This is the phenomenon of \emph{superradiance}\index{superradiance},
first discussed by Zeldovich~\cite{zeld}.
Moreover, a priori, as $\tau_2\to\infty$, this flux can grow without bound.

\subsubsection{More complicated trapping}
We have explained that in the presence of trapped null geodesics,
considerations of geometric optics necessitate a loss 
of derivatives or a degeneration on the trapped set in the putative integrated local energy decay estimate.
In Schwarzschild, this corresponds to the photon sphere $r=3M$, and one could
construct a current ${\bf J}^{X,w}$ where the vector field $X$ vanishes precisely
at $r=3M$ such that 
$(\ref{lowerbnd})$ holds.

As we have seen in Section~\ref{TNG}
in  the Kerr case for $a\ne 0$,  the trapped set has a more complicated structure, which can be fully explained
only in the cotangent bundle picture. Moreover, an argument in \cite{alinhac} shows that
there does not exist a multiplier of type ${\bf J}^{X,w}$ satisftying an analogue
of $(\ref{lowerbnd})$. 

Thus, a classical vector field multiplier to capture trapping is ruled out.

\subsubsection{The coupling}
\label{THECOUPLE}
Recall that the estimate $(\ref{TRUE})$ is not used alone, but is used
in conjunction with the conserved energy estimate $(\ref{Nomos})$.
Thus, even if one can obtain a current of the form $(\ref{oftheformc})$ 
satisfying a nonnegativity property
generalising $(\ref{lowerbnd})$,
it is not at all clear how one obtains an integrated decay estimate,
in view again of superradiance, which now manifests in failure
to control the boundary terms on the right hand side of $(\ref{TRUE})$.
Thus the two difficulties are in fact coupled.

\subsection{The case $|a|\ll M$}
\subsubsection{Boundedness}
\label{boam}
The boundedness problem was resolved in~\cite{dr6}. At the heart of the argument
were the following 2 insights regarding the difficulties of superradiance and trapping:

1.~At least formally, using Carter's separation described in the Introduction, we may decompose
a solution $\psi$ into its superradiant part $\psi_{\flat}$ and its non-superradiant part
$\psi_{\sharp}$
\begin{equation}
\label{decompo}
\psi=\psi_{\flat}+\psi_{\sharp},
\end{equation}
i.e.~where $\psi_{\flat}$ is supported in the frequency range $(\ref{suprange})$ and
$\psi_{\sharp}$ in the  complement of this range.\footnote{The reader may have
already  noticed that it 
is not actually the complete separation which
is necessary to define this particular type of decomposition, 
only Fourier analysis in $t$ and $\phi$. 
See the comments at the end of this section. We suppress here the issue of
justifying this decomposition a priori, since we discuss it in Section~\ref{sec:gen}.}
In the case $|a|\ll M$, one easily sees that {\bf \emph{the superradiant
part $\psi_{\flat}$ is not trapped}}, 
essentially because all future-trapped null
geodesics leave the ergoregion $\mathcal{E}$.
Moreover, one has a more quantitative version of this statement in that one can 
construct a current ${\bf J}^{X,w}$
satisfying an analogue of the non-degenerate
 $(\ref{lowerbnd0})$ at least away from the horizon.

2. One can overcome the difficulty 
of controlling the boundary terms of a ${\bf J}^{X,w}$ current by
the conserved energy,
{\bf \emph{with the help of the red-shift}}. 
We first note that we may write
\begin{align}\label{eq:add}
\int_{\Sigma_{\tau_2}}| {\bf J}^T_\mu[\psi] n^\mu_{\Sigma_\tau}|+\int_{{\mathcal H}^+_{(\tau_1,{\tau_2})}} |{\bf J}^T_\mu[\psi] n^\mu_{\mathcal H}|&\le \int_{\Sigma_{\tau_2}}
 {\bf J}^T_\mu[\psi] n^\mu_{\Sigma_\tau}+\int_{{\mathcal H}^+_{(\tau_1,{\tau_2})}} 
 {\bf J}^T_\mu[\psi] n^\mu_{\mathcal H}\notag\\
  &+ \beta\left (\int_{\Sigma_{\tau_2}} |\pa\psi|^2+ \int_{{\mathcal H}^+_{(\tau_1,{\tau_2})}} |\overline\pa\psi|^2\right)
\\ &\le \int_{\Sigma_{\tau_1}} {\bf J}^T_\mu[\psi] n^\mu_{\Sigma_\tau}+ \beta \left (\int_{\Sigma_{\tau_2}} |\pa\psi|^2+ \int_{{\mathcal H}^+_{(\tau_1,{\tau_2})}} |\overline\pa\psi|^2\right)\notag
\end{align}
where we have used the energy identity for ${\bf J}^T$, and
the fact that $|\partial\psi|^2$, $|\overline\partial\psi|^2$ contain the square sum
of all derivatives appearing  in the fluxes on their corresponding hypersurfaces.
The last two terms in the above inequality 
appear precisely on the left hand side 
of the red-shift estimate \eqref{eq:rshift}.
The constant $\beta$ can in fact be chosen
proportional to the value of the specific angular momentum $a$ of the Kerr solution, which means 
that in the slowly rotating case $|a|\ll M$, given an analogue of
$(\ref{lowerbnd0})$ such that $(\ref{whattodo})$ is bounded by the right hand
side of $(\ref{eq:add})$,
we may add the analogue of
$(\ref{viresti})$ and $(\ref{eq:add})$ and a small amount of 
$(\ref{eq:rshift})$, and the estimates close.

The above is to be applied only to
the superradiant part $\psi_{\flat}$, as it is only for this that we have
an analogue of $(\ref{lowerbnd0})$.  For the non-superradiant part $\psi_{\sharp}$,
we simply apply the argument of Section~\ref{aside}.\footnote{In fact,
the elementary argument of Section~\ref{aside} was first discovered
in the present Kerr context~\cite{dr6}.}

We obtain thus the analogue of $(\ref{actuallybond})$ 
for the solution and thus we retrieve
the Kay--Wald linear stability theorem, but now for Kerr with $|a|\ll M$.

Let us note finally that
the above argument does not require the complete separation
to define $(\ref{decompo})$, but
only the Fourier analysis with respect to      $\partial_t$ and $\partial_\phi$.
Thus, it can be adapted to the wave equation on
more general axisymmetric stationary spacetimes whose metrics are $C^1$-close
to Schwarzschild, and this is in fact the setting of~\cite{dr6}.\footnote{The $C^1$-closeness is what guarantees stability
of positive-definitivity properties of non-degenerate energy currents.} It is interesting
to note that in this very general class, one does not have uniqueness of geodesics, and thus,
the behaviour of the trapped null orbits may be very different from in Schwarzschild. 
In particular, one does not expect to be able to prove quantitative decay
results in this class.

\subsubsection{Decay}
\label{smadec}
For the purpose of decay, the estimate $(\ref{viresti})$
is not a tool for obtaining $(\ref{actuallybond})$ but an end in itself. We must thus obtain the analogue of this estimate
for the whole solution $\psi$, not just the surperradiant projection $\psi_{\flat}$. 

As we have seen, for this, in view of~\cite{alinhac},
one must go beyond classical multiplier currents. 
The property of trapped null geodesics discussed at the end of Section~\ref{TNG}  
strongly suggests that suitable ${\bf J}^{X,w}$ currents could be
constructed if one were to microlocalise the construction, to distinguish between
geodesics with different sets of conserved quantities, and degenerate accordingly
at different values of $r$ associated to trapped null orbits, all (in the case
$|a|\ll M$!) close to $3M$.

Again, Carter's separation\footnote{Unlike in the  case of the boundedness property
of Section~\ref{boam} above,
here, as will be apparent we really need the complete separation.} 
into modes $(\ref{modetype})$ suggests a particularly 
 natural microlocalisation suitable for this problem. 
In fact, the existence of this separation is equivalent to the separability
of geodesic flow and thus, it should naturally capture properties of geodesic flow.
Using this  microlocalisation we shall indeed be able to construct microlocal
versions of the currents ${\bf J}^{X,w}$ ensuring the validity of the
analogue of $(\ref{lowerbnd})$.
We shall see how this is done in Sections~\ref{sec:gen}--\ref{sec:large} below.

The above microlocalisation was first carried out in Chapter 5.3 of~\cite{jnotes}. 
We note that, independently, \cite{tattoh} showed
a similar estimate in the $|a|\ll M$ case considered here,
 using a different microlocalisation relying on
the general machinery of pseudo-differential operators. The integrability
of geodesic flow is essential here too.
Later, a third `microlocalisation' was achieved in~\cite{andB}  by commutation 
with a higher order differential operator constructed from the Carter tensor.
This attractive approach allows for a
formulation of the whole argument in terms of differential (as opposed to
pseudodifferential) operators, at the expense of showing results only for
higher energy norms with strong weights at infinity.

Given the fact that boundedness has been proven (cf.~Section~\ref{boam}), 
one need not worry
about boundary terms and the difficulty of Section~\ref{THECOUPLE}. 
This yields the analogue of the integrated decay statements
$(\ref{TRUE0})$, $(\ref{TRUE})$, where the degeneration is now
on a set containing the physical space projection of the trapped set. Alternatively, 
one can apply directly part~2.~from Section~\ref{boam}   to obtain 
the boundedness of the boundary terms,
bypassing part~1. As we shall see, however, the insight of part 1.~of
Section~\ref{boam} will return
with a vengeance in the study of the general $|a|<M$ case!

\subsection{The general case $|a|< M$}
\label{edwgeviko}
The nature of geodesic flow as described in Section~\ref{TNG}
suggests that `trapping' is qualitatively similar for the whole
range of   $a$, and thus, microlocalised currents should 
again be able to capture this difficulty. 
This is indeed the case and is particularly easy to see with 
respect to the microlocalisation based on Carter's separation, as the latter is so tightly
connected with the underlying geometry.

The real difficulty now, however, is trying to apply $(\ref{eq:add})$.  For in the general case
$|a|<M$, the coupling constant $\beta$ is no longer small, and thus the extra terms
 cannot be
simply absorbed by adding a small amount of the red-shift estimate. 
One has to identify some non-trivial \emph{structure} in the coupling of the various components.

The miracle that saves the day is the following: The insight 1.~of Section~\ref{boam}
in fact holds for the entire $|a|<M$ range! That is to say, superradiant frequencies
are not trapped, in fact, we can quantify this by saying that
if the frequency is sufficiently large, one can construct
a ${\bf J}^{X,w}$-type current with ${\bf K}^{X.w}$ 
sufficiently positive to accommodate coupling with 
$(\ref{eq:add})$ for an arbitrary $\beta$.

This miracle is somewhat unexpected because it is certainly {\bf \emph{not}} the case that
all trapped null geodesics leave the ergoregion $\mathcal{E}$. It would be interesting
to understand better the underlying reasons (if any) behind
 the phase-space disjointness of
the trapping and superradiance phenomena. If better understood, in 
view of the programme of~\cite{alexakis},
this could have bearing on the uniqueness of Kerr problem.

The above remarks all concern the situation where (at least some) frequencies
are large. One must also deal with the case where all frequencies are bounded without
there being a large frequency parameter to exploit.
In the $|a|\ll M$ case, these can be dealt with by a perturbation argument from
Schwarzschild.\footnote{This was the approach adopted in~\cite{jnotes}. In 
our alternative approach~\cite{dr7}, we deal with the low frequencies 
as in the present paper (see Section~\ref{sec:small}), precisely because the method is more systematic,
can be used in the general $|a|<M$ case, and yields an independent self-contained
proof
of the Schwarzschild result.}
Again, however,
the separation gives a systematic way to deal with these difficulties which
is important in the large $|a|<M$ case where perturbation arguments from Schwarzschild
are not available.  See Sections~\ref{sec:small} and Section~\ref{whit}.

\section{Carter's separation as a geometric microlocalisation}   
\label{ayrilik}
We have given a complete overview of the proof, and now we begin with
the technical details. As we have explained in the previous Section, 
we shall appeal to Carter's separation of the wave
equation, viewing it  as a method to frequency-localize energy estimates
in a manner particularly
suited to the local and global geometry of Kerr. In the present section,
we introduce the relevant formalism.

\subsection{Carter's separation}
\label{karteri}
The separation of $\Box_g\psi=0$ requires taking the Fourier transform in $t$, and
then expanding into what are known as oblate spheroidal harmonics. 
As we shall describe in Section~\ref{sec:gen}, we will
apply this separation to solutions of the inhomogeneous equation 
\begin{equation}
\label{INHOMOG}
\Box_g\Psi = F,
\end{equation}
where $\Psi$ is related to $\psi$ by the application of a suitable cutoff, making
$\Psi$ then $L^2(t)$ for each fixed $(r,\theta,\phi)$, where these denote Boyer--Lindquist
coordinates on $\mathcal{R}$.

We may thus write $\Psi(t,r,\theta,\phi)$ in the form
$$
\Psi(t,r,\theta,\phi)=\frac{1}{\sqrt{2\pi}}\int_{-\infty}^\infty e^{-i\omega t} \widehat\Psi (\omega,r,\theta,\phi) d\omega
$$
and further decompose 
\[
\widehat\Psi(\omega, r,\theta,\phi) = \sum_{m\ell} \Psi^{(a\omega)}_{m\ell}(r)
S_{m\ell}(a\omega,\cos \theta)e^{im\phi}.
\]
The collection $\{S_{m\ell}(\xi,\cos \theta)e^{im\phi}\}_{m\ell}$ of the eigenfunctions of the 
self-adjoint operator 
\[
P(\xi)\, f= -\frac 1{\sin\theta} \frac{\partial}{\partial\theta} \left (\sin\theta \frac{\partial}{\partial\theta}f\right)
-\frac{\partial^2 f}{\partial\phi^2}\frac{1}{\sin^2\theta}
- \xi^2 \cos^2\theta f.
\]
on $L^2(\sin\theta\, d\theta\, d\phi)$ forms a complete orthonormal basis on 
$L^2(\sin\theta\, d\theta\, d\phi)$. It is parametrized by a set of real discrete eigenvalues $\lambda^{(\xi)}_{m\ell}$
\[
P(\xi)\, S_{m\ell}(\xi,\cos \theta)e^{im\phi}=\lambda_{m\ell}(\xi) S_{m\ell}(\xi,\cos \theta)e^{im\phi},
\]
which have the property that
\begin{equation}\label{eq:lam}
\lambda_{m\ell}{(\xi)}+\xi^2\ge |m|(|m|+1).
\end{equation}
For $\xi=0$ the {\it oblate spheroidal harmonics} $S_{m\ell}(\xi,\cos \theta)e^{im\phi}$ 
reduce to the 
standard spherical harmonics $Y_{m\ell}$ with the corresponding 
eigenvalues $\lambda{(0)}_{m\ell}=\ell(\ell+1)$.

The usefulness of the above decomposition 
in the context of the wave equation on Kerr is given by the 
following result of Carter~\cite{cartersep2}. If  $\Psi(t,r,\theta,\phi)$ 
satisfies
\[
\Box_g\Psi = F,
\]
for $g=g_{a,M}$ a Kerr metric with $|a|<M$, 
with $\Psi(t,\cdot)$ in $L^2(t)$,
then defining the coefficients 
$\Psi_{m\ell}^{(a\omega)}(r)$, $F_{m\ell}^{(a\omega)}(r)$ as above,
the following holds:
\begin{align}
\label{CartersODE}
\nonumber
\Delta \frac{d}{dr} \left (\Delta \frac{d\Psi_{m\ell}^{(a\omega)}}{dr}\right)& + \left (a^2m^2 + (r^2+a^2)^2\omega^2-\Delta (\lambda_{m\ell}+a^2\omega^2) \right) \Psi_{m\ell}^{(a\omega)}\\
&=(r^2+a^2)\Delta\,
F_{m\ell}^{(a\omega)}.
\end{align}
Using the definition $(\ref{r*def})$ of $r^*$ and setting
\[
u^{(a\omega)}_{m\ell}(r)=(r^2+a^2)^{1/2}
 \Psi^{(a\omega)}_{m\ell} (r),\qquad 
 H^{(a\omega)}_{m\ell}(r)=\frac{\Delta F^{(a\omega)}_{m\ell}(r)}{(r^2+a^2)^{1/2}}.
\]
we obtain
\index{$\Psi$-related! $u$ (suitably rescaled version of $\Psi^{(a\omega)}_{m\ell} (r)$)}
\index{$\Psi$-related! $H$ (suitably rescaled version of $ F^{(a\omega)}_{m\ell}(r)$)}
\begin{equation}
\label{e3iswsntouu}
\frac{d^2}{(dr^*)^2}u^{(a\omega)}_{m\ell}+(\omega^2 - V^{(a\omega)}_{m\ell }(r))u = 
H^{(a\omega)}_{m\ell}
\end{equation}
where
\[
V^{(a\omega)}_{m \ell}(r)= \frac{4Mram\omega-a^2m^2+\Delta (\lambda_{m\ell}+\omega^2a^2)}{(r^2+a^2)^2}
+\frac{\Delta(3r^2-4Mr+a^2)}{(r^2+a^2)^3}
-\frac{3\Delta^2 r^2}{(r^2+a^2)^4}.
\index{fixed functions! spacetime functions! $V$ (potential $V=V^{(a\omega)}_{m \ell}(r^*)$ arising from separation)}
\]
In the the Schwarzschild case:
\begin{equation}
\label{at0freq}
V^{(0\omega)}_{m\ell}(r) = (r-2M)\left(\frac{\lambda_{m\ell}}{r^3}+\frac{2M}{r^4}\right),
\end{equation}
\begin{equation}
\label{at0freq'}
\left(\frac{dV}{dr^*}\right)^{(0\omega)}_{m\ell}(r)= \frac{r-2M}{r}\left(\frac{2\lambda_{m\ell}(3M-r)}{r^4}+\frac{2M(8M-3r)}{r^5}\right).
\end{equation}
We note that in this picture the application of $\frac{\pa}{\pa t}$ to $\Psi$ corresponds to the multiplication
of $\Psi^{(a\omega)}_{m\ell}$ by $i\omega$, while $\frac{\pa}{\pa\phi}$ corresponds to the multiplication by $i m$.

\subsection{Notational conventions}
In what follows, let us  suppress the dependence of
$u$, $H$ and $V$ on $a\omega$, $m$, $\ell$ in our notation.\footnote{Before suppressing
the dependence, it might be useful to remark that in fact, the functions
$u$, $H$, $V$, etc., depend on both $a\omega$ and $a$, whereas the
parameters $\lambda_{m\ell}$ depend only on $a\omega$. For the former 3, will view
reference to the $a$-dependence as implicit in the reference to $a$ in $a\omega$. Hence,
our writing $0\omega$ instead of $0$ in $(\ref{at0freq})$.}

We will also use the notation $'=\frac{d}{dr^*}$ and introduce a new useful non-negative parameter
$$
\Lambda=\lambda_{m\ell}(a\omega)+a^2\omega^2,
$$
which according to \eqref{eq:lam} obeys
$$
\Lambda\ge |m|(|m|+1).
$$

A word of warning: For fixed $g_{M,a}$, we 
shall often refer to $r^*$-ranges by their corresponding
$r$ ranges, and functions appearing in most estimates will be written in terms of $r$. 
Moreover, given an $r$-parameter such as $R$, then $R^*$ will denote $r^*(R)$.
It is important to remember at all times that $'$ always means $\frac{d}{dr^*}$!

\subsection{Boundary conditions}
As we shall describe in Section~\ref{sec:gen}, 
the separation process will be applied to $\Psi=\xi(\tau) \psi$ obtained 
by applying a cut-off $\xi$ with the property that $\xi(\tau)=1$ for $\tau\ge 1$ and $\xi(\tau)=0$
for $\tau\le 0$ to the solution of the homogeneous wave equation $\psi$. As a consequence, 
the coefficients $\Psi^{(a\omega)}_{m\ell}(r)$ and hence $u^{(a\omega)}_{m\ell}(r)$ will
satisfy the {\it outgoing boundary conditions}:
\begin{align}
&u'+i(\omega-\frac {a m}{2Mr_+})u=0,\qquad r=r_+,\label{eq:b-}\\
&u'-i\omega u =0,\qquad\qquad\qquad\quad r=\infty \label{eq:b+}.
\end{align}
They correspond to the standard outgoing conditions on $\Psi$
\begin{align*}
&\left (\frac {\pa}{\pa r^*}-\left(\frac {\pa}{\pa t}+\frac {a}{2Mr_+}\frac {\pa}{\pa\phi}\right)\right )\Psi=0,\qquad r=r_+,\\
&\left (\frac {\pa}{\pa r^*}+\frac {\pa}{\pa t}\right)\Psi =0,\qquad\qquad\qquad\qquad r=\infty.
\end{align*}
We note that the vectorfield $L=\frac {\pa}{\pa t}+\frac {a}{2Mr_+}\frac {\pa}{\pa\phi}$ is the distinguished 
null generator of the future horizon ${\mathcal H}^+$ mentioned, for instance, in the earlier discussion 
of the red-shift effect and surface gravity.

\subsection{The separated current templates}
\label{sct}

To describe the analogue of energy currents localized to
frequency triplet
$(\omega,m,\ell)$, it will be convenient to define the following current templates.

\subsubsection{The microlocal ${\bf J}^T$}
We first define the microlocal analogue of the conserved energy current, i.e.~the 
${\bf J}^T$ current:
$$
Q_T[u]=\omega\, {\text Im}  (u'\overline {u}).
$$
We compute
$$
Q_T'=\omega\, {\text Im}  (H\overline {u}).
$$
We observe that the boundary conditions \eqref{eq:b-}, \eqref{eq:b+} imply that  
$$
Q_T(r_+)=-\omega (\omega-\frac {am}{2Mr_+})|u|^2,\qquad Q_T(\infty)=\omega^2 |u|^2
$$
and that $Q_T(\infty)$ and $-Q_T(r_+)$ are both positive in the {\it non super-radiant regime}
$$
\frac {m\omega}{a}\not\in [0,\frac {m^2}{2Mr_+}].
$$

\subsubsection{The microlocal virial currents ${\bf J}^{X,w}$}
Next, we  define the microlocal analogue of the virial currents ${\bf J}^{X,w}$
where $X$ is in the direction of $\partial_{r^*}$, and $w$ is a suitable function.

It is most natural to parametrize such choices of currents as follows.
For arbitrary functions $f(r^*), h(r^*), y(r^*)$, we  define the currents
\begin{align*}
&Q^f_0=f\left [|u'|^2+(\omega^2-V)|u|^2\right]+f'\, {\text Re} (u'\overline {u})-\frac 12 f'' |u|^2,\\
&Q^h_1=h\,{\text{Re}} (u'\overline{u}) -\frac 12 h' |u|^2,\\
&Q^y_2=y \left [|u'|^2+(\omega^2-V) |u|^2\right].
\end{align*}
We note the simple relation: $Q_0^f=Q_1^{f'}+Q_2^f$.
We compute 
\begin{align*}
&(Q^f_0)'=2f'|u'|^2-fV' |u|^2 -\frac 12 f{'''}|u|^2+2 f \,{\text{Re}} (u'\overline{H}) + f'\, {\text{Re}} (u\overline{H}),\\
&(Q^h_1)'=h \left [|u'|^2+(V-\omega^2) |u|^2\right]-\frac 12 h'' |u|^2 + h\, {\text{Re}} (u\overline {H}),\\
&(Q^y_2)'=y' \left [|u'|^2+(\omega^2-V) |u|^2\right]-yV' |u|^2+y \,{\text{Re}} (u'\overline{H}).
\end{align*}

\subsubsection{The microlocal red-shift current ${\bf J}^{N}$}
Finally, we define the microlocal analogue of the red-shift current  ${\bf J}^N$.
This takes the form
$$
Q^z_{\text {red}}=z\left [|u'+i(\omega-\frac{am}{2Mr_+}) u|^2+(\omega^2-V-|\omega-\frac{am}{2Mr_+}|^2) |u|^2\right],
$$
where $z$ will be chosen to be a suitable function blowing up as $r\to r_+$.

For convenience we denote
$$
\omega_+:=\frac{am}{2Mr_+},
$$ 
a constant which can be also characterized by the relation
$$
\omega^2-V(r_+)=|\omega-\omega_+|^2.
$$
Let 
$$
\tilde V:=V+|\omega-\omega_+|^2-\omega^2
$$
so that $\tilde V(r_+)=0$ and $\tilde V'(r)=V'(r)$. We then compute
\begin{equation}\label{eq:rs}
(Q^z_{\text{red}})'=z'|u'+i(\omega-\omega_+)u|^2 - \left (z\tilde V\right)'|u|^2 + 2z\, {\text{Re}}\,(u'\overline H).
\end{equation}

\section{Proof of Theorem~\ref{thrm:state}}\label{sec:gen}

In the rest of this paper, we give detailed account of the proof of
Theorem~\ref{thrm:state}. 
\subsection{The continuity argument in $a$}
The proof of Theorem~\ref{thrm:state} proceeds via 
a continuity argument in the
Kerr parameter $a$.  This is central to the
logic of the proof because we are relying (recall Section~\ref{karteri}) 
on taking the Fourier transform
in time, and this requires some a priori information. 

We must show thus that the set of $a$ for which the statement of
Theorem~\ref{thrm:state} holds (with a constant $C=C(a)$ whose
form is given by the proof) is nonempty, and both open and closed
in the subset topology of $(-M,M)$. We note that in principle,
the function  $C(a)$ will diverge as $|a|\to M$.

The value $a=0$ corresponds to Schwarzschild. The validity of the result
 in the Schwarzschild
case thus satisfies the requirement on non-emptiness. (In fact, the considerations
of Section~\ref{sec:small} below can be used to provide an independent, self-contained proof
for the Schwarzschild case. See~\cite{dr7}.)

\subsection{Openness}
The template for the openness argument is actually the small $|a|\ll M$ case, treated
in detail in~\cite{dr7}, which can be thought of as implementing the openness argument
in the vicinity of $a=0$.

The actual openness argument is an amalgation of techniques from the closedness
argument and the small $a$ case. For convenience, here we shall 
treat the openness case as if it were exactly the small $a$ case. This will allow us
also to review in more detail this special case.

\subsubsection{Cutoffs}
For openness, we do not know that $\psi$ decays a priori  in time, and thus, our
cutoff $\xi$ must be supported only in $0\le \tau\le \tau_f$ for some large $\tau_f$,
where $\xi=1$ in the range $1\le \tau\le \tau_f-1$ say.

By a bootstrap argument, $\tau_f$ will finally be taken to infinity. 
The proximity to something known (inherent in proving openness) will be essential
here. See Section~\ref{sumfopen}.

\subsubsection{A microlocalised virial current}
For convenience, let us actually restrict to the $|a|\ll M$ case in our discussion of
this construction, because it is a nice illustration of how proximity
to a solution where estimates have already been established can be used.
(The general openness will follow from an amalgamation of the ideas
here and in the closedness section.) 

We have discussed   in Section~\ref{smadec} 
that we need a microlocal version
of the current ${\bf J}^{X,w}$  and of the positivity
statement which schematically we write
\[
\int{\bf K}^{X,w}[\Psi]\ge b\int_{R_0}^{R_1} |\partial_{r^*}\Psi|^2 +  \chi |\partial\Psi|^2+|\Psi|^2
-\int \mathcal{E}^{X,w}[\Psi].
\]
where $r_+<R_0<3M<R_1<\infty$. 
The $\mathcal{E}^{X,w}[\Psi]$ occurs (recall $(\ref{eq:Not})$) 
from the fact that $\Psi=\xi\psi$
and thus satisfies an inhomogeneous wave equation.

In the formalism of our separation,
this means that for various frequency ranges $(\omega, m, \ell)$, 
by employing various compositions of the virial microlocalised currents
$Q_0^f$, $Q_1^h$, $Q_2^y$, we seek a current $Q$ such that
\begin{equation}
\label{yaxvw}
\int_{-\infty}^{\infty}  Q'[u] \ge \int_{R_0^*}^{R_1^*}  |u'|^2 +\chi (\lambda +\omega^2)|u^2| +|u|^2   - \int_{-\infty}^{\infty}
H \cdot  (f, h, y) \cdot (u, u'),
\end{equation}
where the last term is schematic for various combinations (such as $f {\rm Re} (u'\overline H)$)
which occur.  The terms $Q$, $u$, $\chi$, $\lambda$, $H$, $f$, $h$, $y$ should
all be understood with the full set of indices, i.e.~$u^{(a\omega)}_{m\ell}$, etc.,
here omitted for brevity.
For frequencies for which trapping
is relevant, $\chi$ will vanish at a unique value of $r$ associated with the
frequency triple. 
One restricts the domain of integration on
the first term on the right hand side because one expects this virial current
not to control things at the horizon and infinity.

We may rewrite $(\ref{yaxvw})$ above as an estimate
\[
b \int_{R_0^*}^{R_1^*}  |u'|^2 +\chi (\lambda +\omega^2)|u^2| +|u|^2 \le 
\int_{-\infty}^{\infty}  Q'[u] + \left| \int_{-\infty}^{\infty}
H \cdot  (f, h, y) \cdot (u, u')\right|,
\]
and  thus 
\begin{equation}
\label{modred}
b \int_{R_0^*}^{R_1^*}  |u'|^2 +\chi (\lambda +\omega^2)|u^2| +|u|^2 \le 
 Q[\infty]-Q[-\infty] + \left| \int_{-\infty}^{\infty}
H \cdot  (f, h, y) \cdot (u, u')\right|.
\end{equation}

\subsubsection{Summing}
\label{sumfopen}

Restoring now the indices, summing $(\ref{modred})$ over
$m$ and $\ell$ and integrating over $\omega$, we obtain by Plancherel for $r_c\le R_0<R_1\le R_c$
\begin{align}\label{eq:sumearlier}
b \int_0^{\tau_f}
 d\tau \int_{\Sigma_\tau\cap \{r\in [R_0,R_1]\}}  \left (\chi |\pa\psi|^2 +|\partial_{r^*}\psi|^2+ |\psi|^2\right)&\le
 \int_{{\mathcal H}^+} |{\bf J}^T_\mu n^\mu_{{\mathcal H}}| +
 \int_{{\mathcal I}^+} {\bf J}^T_\mu n^\mu_{{\mathcal I}} \notag\\
 &+\int_{\Sigma_0} |\pa\psi|^2+\int_{\Sigma_{\tau_{f}-1}} |\pa\psi|^2.
\end{align} 
Here, $\chi$ now denotes a smooth cut-off function which vanishes
on a small ($a$-size) neighborhood of $r=3M$.

Below we explain the main features of how we obtained \eqref{eq:sumearlier}:

1.~The range of $\tau$-integration on the left hand side of \eqref{eq:sum} is restricted to $[0,\tau_f]$.
This is due to the fact that our microlocal analysis was applied to the function $\Psi=\xi(\tau) \psi$
with a cut-off function $\xi$ vanishing for $\tau\le 0$, $\tau\ge \tau_f$. 
The former property of $\xi$ led to the microlocal outgoing 
boundary conditions~\eqref{eq:b-}, \eqref{eq:b+}.

2.~The boundary terms on ${\mathcal H}^+$ and ${\mathcal I}^+$ arise from the corresponding 
boundary terms $Q(\pm\infty)$ in the microlocal virial inequalities. The fact that these terms
are controlled by the ${\bf J}^T$ fluxes is a consequence of the boundary conditions \eqref{eq:b-}, \eqref{eq:b+}.
For instance, 
$$
Q_0^f(+\infty)= f (|u'|^2+(\omega^2-V)|u|^2)(+\infty)=2f\omega^2 |u|^2(+\infty)= if \omega \overline u (-i\omega u -u') (+\infty),  
$$
which by Plancherel can easily be seen to be bounded by the 
${\bf J}^T$-flux upon summation.

3.~The contribution of the inhomogeneous term $H$ in $(\ref{modred})$  gives rise to the terms integrated over $\Sigma_0$ and $\Sigma_{\tau_{f}-1}$ on the right hand side
of \eqref{eq:sumearlier}. Let us see briefly how this occurs.

Recall that $H$ arises from the term $F=\Box_g \xi\, \psi + 2\nab^\mu\xi\, \nab_\mu \psi$, introduced 
by the cut-off $\xi$ mutliplied to $\psi$, and that
$F$ is thus  supported in 
$\tau\in [0,1]\cup[\tau_f-1,\tau_f]$.

For concreteness, consider a term
\[
\int f^{(a\omega)}_{m\ell}(r)
{\text{Re}} (\frac {d u^{(a\omega)}_{m\ell}}{dr^*}\overline H^{(a\omega)}_{m\ell})
\]

Were this integral in fact restricted to $[R_0,R_1]$, it is easy to see that
it can be controlled by a little bit of the left hand side of $(\ref{eq:sumearlier})$
and the last two terms on the right hand side. 
For we apply Cauchy--Schwarz to bound this term by
\[
\int{ |f^{(a\omega)}_{m\ell}(r)|\left(
\left |\frac {d u^{(a\omega)}_{m\ell}}{dr^*}\right|^2+\epsilon^{-1} |H^{(a\omega)}_{m\ell}|^2\right)},
\]
and applying Plancherel, one can then absorb the first term into the left
hand side, while the right hand side, together with a finite time 
energy estimate and the above comment on the support of $F$, can easily
be seen to be bounded by the claimed boundary terms.

On the complement of $[R_0,R_1]$, one needs a slightly 
different strategy, which takes into account that the term $H^{(a\omega)}_{m\ell}(r)$ arises from the
compactly supported expression 
$F=\Box_g \xi(\tau) \Psi + 2 \nab^\mu \xi(\tau)\, \nab_\mu \Psi$ with a smooth cut-off function $\xi(\tau)=1$ for $\tau\ge 1$ and 
$\xi(\tau)=0$ for $\tau\le 0$. In \cite{dr7}, 
we have arranged our mutipliers, i.e., function $f^{(a\omega)}_{m\ell}(r)$ in this 
example to be independent of the parameters $\omega, m, \ell$ outside of $(R_0,R_1)$ so that after summation 
in $m, \ell$ and integration in $\omega$ we obtain the term 
$$
\int_0^{\infty} d\tau \int_{\Sigma_\tau\cap\{r\not\in (R_0,R_1)\}} \frac {d\Psi}{dr^*} F.
$$
Since $F$ is supported in the interval $\tau\in (0,1)$ we can easily relate the above expression to the initial data.
More generally, the basic idea for treating compactly supported in $\tau$ inhomogeneous terms in the equation
$\Box_g\psi=F$ is to couple them with the ${\bf J}^T$ based energy estimate for $\psi$. (The precise analog of
this for the wave equation in Minkowski space is that the integrated local energy decay estimate can have 
two incarnations (with arbitrary $\alpha>0$)
\begin{align*}
&b\left (\|\frac {\pa\psi}{\langle r\rangle^{\frac 12+\alpha}}\|_{L^2_tL^2_x} + \|\frac {\psi}{\langle r\rangle^{\frac 32+\alpha}}\|_{L^2_tL^2_x}\right) \le \|\pa\psi_0\|_{L^2_x} + \|\langle r\rangle^{\frac 12+\alpha} F\|_{L^2_tL^2_x},\\
&b\left (\|\frac {\pa\psi}{\langle r\rangle^{\frac 12+\alpha}}\|_{L^2_tL^2_x} + \|\frac {\psi}{\langle r\rangle^{\frac 32+\alpha}}\|_{L^2_tL^2_x}\right) \le \|\pa\psi_0\|_{L^2_x} + \|F\|_{L^1_tL^2_x}
\end{align*}
and the second estimate is more advantageous if $F$ has compact support in $t$.)

These technical issues are carefully treated in \cite{dr7} and we proceed without making any further comments.

\subsubsection{Completing the argument}
\label{compfopen}

We may now couple $(\ref{eq:sumearlier})$ to $(\ref{eq:add})$, and, applying the red-shift
estimate,
as described in Section~\ref{boam}, we may obtain
\begin{equation}
\label{komble}
\int_0^{\tau_f}d\tau \int_{\Sigma_\tau \cap\{r\le R_1\}}
(\chi |\partial\psi|^2 +|\psi|^2 )\le C \int_{\Sigma_0} |\partial \psi|^2,
\end{equation}
where we are using that we may treat $\beta$ as a small parameter to absorb terms.
The cutoff time was arbitrary and does not appear in the constant, and thus,
we may finally replace $\tau_f$ with $\infty$.

We have thus reduced openness to showing
$(\ref{modred})$ for all frequency ranges.
This will be done explicitly in    Section~\ref{sec:small}
for  $|a|\ll M$, giving us the opportunity
to review this simpler case--treated in full detail in~\cite{dr7}--first. As said earlier,
the actual openness
argument for the general case can be implemented by an amalgamation of
these multipliers and the types of multipliers we shall use in the closedness case.
Thus, the real new difficulties of the general $|a|<M$ case concern closedness.

\subsection{Closedness}
For closedness, it suffices to improve an estimate 
for a solution already known to decay,
i.e.~it suffices to consider a sequence of $a_i\to a$ for which the statement of
Theorem~\ref{thrm:state} holds, and to show that the constant $C(a_i)$ of the theorem can 
be uniformly bounded as $a_i\to a$.

\subsubsection{Cutoffs}
In view of the fact that we already know that the solution is integrable towards the future,
we need only cutoff  near the initial data surface $\Sigma_0$. That is to say, we
apply a cutoff $\xi$ to the solution $\psi$ of $(\ref{eq:box})$
such that $\xi=1$ for $\tau\ge 1$ and $\xi=0$ for $\tau\le 0$ 
and we then apply the separation of Section~\ref{karteri} to $\Psi=\xi\psi$.
In contrast to the openness case, in view of the new cutoff,
the inhomogeneous term $H^{(a\omega)}_{m\ell}(r)$ can be thought to arise 
solely from initial data for $\psi$ on $\Sigma_0$.

\subsubsection{A microlocalised combined current}
Whereas in the small $|a|\ll M$ case, we needed merely to microlocalise
the construction of ${\bf J}^{X,w}$, and the boundary terms could be handled
using $(\ref{eq:add})$ and the smallness of the parameter $\beta$, 
here, the coupling between the virial current, the problem of superradiance, 
and the red-shift is nontrivial. To exploit the structure described in Section~\ref{edwgeviko},
{\bf\emph{all currents must be microlocalised}}.
Essentially, we would like a microlocal version of the inequalities
\[
\int {\bf K}^{X,w}[\Psi]+{\bf K}^{\alpha N}[\Psi]+{\bf K}^{\beta T}[\Psi] \ge  \int |\partial_{r^*} \Psi |^2+ \chi|\pa \Psi |^2 +\Psi^2-\int \mathcal{E}^{X,w}[\Psi],
\]
\[
\int_{\mathcal{H}^+ \cup \mathcal{I}^-} ({\bf J}_\mu^{X,w}+\beta{\bf J}_\mu^{N}+\alpha{\bf J}_\mu^{T}) n^\mu \le 0,
\]
where, not only the choice of $X$ but also the coupling constants $\alpha$, $\beta$
are `microlocalised'. To put it another way, using the linearity of the currents,
we seek a microlocalised version of the combined currents:
\[
{\bf J}^{X+\beta N+\alpha T, w}_\mu.
\]

In our formalism (see Section~\ref{sct}), what we are seeking, for each
frequency triple $(\omega, m, \ell)$, a combination $Q$
 of $Q_T$, $Q^z_{\rm red}$ and
the virial currents $Q_0$, $Q_1$, $Q_2$   such that $(\ref{yaxvw})$ holds and moreover
\[
Q[\infty]-Q[-\infty]\ge 0
\]
so we may drop the boundary terms and write
\begin{equation}
\label{closedredu}
b \int_{R_0^*}^{R_1^*}  |u'|^2 +\chi (\lambda +\omega^2)|u|^2 +|u|^2 \le 
\left| \int_{-\infty}^{\infty}
H \cdot  (f, h, y) \cdot (u, u')\right|.
\end{equation}

\subsubsection{Summing and completing the argument}
\label{sum&comp}
If we succeed in obtaining $(\ref{closedredu})$ for all frequency triples, then
as in Section~\ref{sumfopen}, we may sum, where now we obtain simply: 
\begin{align}\label{eq:sumearliercl}
b \int_0^{\infty}
 d\tau \int_{\Sigma_\tau\cap \{r\in [R_0,R_1]\}}  \left (\chi |\pa\psi|^2 +|\partial_{r^*}\psi|^2+ |\psi|^2\right)&\le
\int_{\Sigma_0} |\pa\psi|^2.
\end{align} 
The difficulty of the $\Sigma_{\tau_f}$ integral does not occur now in view
of the new definition of the cutoffs. The rest of the proof follows as in Section~\ref{compfopen},
completing the argument for closedness.

We have thus reduced closedness to showing
$(\ref{closedredu})$ for all frequency ranges.
We shall show this in Section~\ref{sec:large}, completing thus the proof
of Theorem~\ref{thrm:state}.

\section{The small $a$ case}\label{sec:small}
This section and the next will contain the microlocalised estimates
appealed to in the proof of Theorem~\ref{thrm:state} as given in Section~\ref{sec:gen}.
We begin by partitioning the frequency space into relevant frequency ranges.

\subsection{The frequency ranges}
\label{freerange}
Let $\omega_1$\index{fixed parameters! small parameters! $\omega_1$ (frequency parameter)}, 
$\lambda_1$\index{fixed parameters! large parameters! $\lambda_1$ (frequency parameter)}
be (potentially large) parameters
to be determined, and $\lambda_2$\index{fixed parameters! small parameters! $\lambda_2$ (frequency parameter)}
be a (potentially small) parameter to be determined.
We define the frequency ranges
$\mathcal{F}_{\mbox{$\flat$}}$\index{frequency ranges! $\mathcal{F}_{\mbox{$\flat$}}$ (frequency triplet range,
`bounded' frequencies)}, 
$\mathcal{F}_{\lessflat}$\index{frequency ranges! $\mathcal{F}_{\lessflat}$ (frequency triplet range,
`angular-dominated' frequencies)},
$\mathcal{F}_{\mbox{$\natural$}}$\index{frequency ranges! $\mathcal{F}_{\mbox{$\natural$}}$ (frequency triplet range, `trapped' frequencies)},
 $\mathcal{F}_{\mbox{$\sharp$}}$\index{frequency ranges! $\mathcal{F}_{\mbox{$\natural$}}$ (frequency triplet range, `time-dominated' frequencies)} 
 by
\begin{itemize}
\item
$\mathcal{F}_{\mbox{$\flat$}}=\{(\omega, m, \ell)$ :
$|\omega|\le \omega_1$, $\lambda_{m\ell}(a\omega) \le \lambda_1\}$
\item
$\mathcal{F}_{\lessflat}=\{(\omega, m, \ell)$ :
$|\omega|\le \omega_1$, $\lambda_{m\ell}(a\omega) >\lambda_1\}$
\item
$\mathcal{F}_{\mbox{$\natural$}}=\{(\omega, m, \ell)$ :
$|\omega|\ge \omega_1$, $\lambda_{m\ell}(a\omega) \ge \lambda_2\omega^2\}$
\item
$\mathcal{F}_{\mbox{$\sharp$}}=\{(\omega, m, \ell)$ :
$|\omega|\ge \omega_1$, $\lambda_{m\ell}(a\omega) <\lambda_2\omega^2\}$.
\end{itemize}
Our constructions will have a separate philosophy for each of the above ranges.

\subsection{The $\mathcal{F}_{\mbox{$\flat$}}$ range (bounded frequencies)}
\label{elow}
This is a compact frequency range and, in view of the fact that
we have already constructed multipliers in the Schwarzschild case (see~\cite{dr3, dr5}),
by stability considerations, their positivity properties carry over for $|a|\ll M$. See our
proof in~\cite{jnotes}.

In \cite{dr7}, we chose to give a different argument, which, on one hand, allowed us to treat 
the general axi-symmetric case $(|a|<M, m=0)$, and furthermore, is more useful in applications
to the openness regime of the general case.

Let us define 
\[
V_{\rm new}=\frac{4Mram\omega-a^2m^2}{(r^2+a^2)^2},
\]
\begin{eqnarray}
\label{plumi}
\nonumber
V_+ &=& V-V_{\rm new}\\
\nonumber
	&=& \frac{\Delta}{(r^2+a^2)^{4}}
\left( (\lambda_{m\ell}+\omega^2a^2)(r^2+a^2)^2+ (2Mr^3+a^2r^2+a^4-4Mra^2) \right).
\\
\end{eqnarray}
We have
\[
V_+> 0
\]
for $r>r_+$. Moreover, according to our conventions, for all $|a|\le a_s\ll M$,
\[
B(\Delta/r^2) \Lambda r^{-2}  +B r^{-3}\ge
V_+ 
\ge (\Delta/r^2) b \Lambda r^{-2}  +b(\Delta/r^2) r^{-3}.
\]

We may now fix an
arbitrary $r_++s_-> r_{c}>2M$, where $s_-$ is as in \eqref{eq:rshift} and can be chosen only to depend on $a_s$,
and it follows that for sufficiently small $a_s$ depending in particular on this choice,
we have
for
 $|a|\le a_s$ and $(\omega,m,\ell)\in \mathcal{F}_{\mbox{$\flat$}}$ 
 the following inequality in the region $r\ge r_c$:
\begin{equation}
\label{Vpos1}
B(\Delta/r^2) \Lambda r^{-2}  +B r^{-3}
\ge 
V 
\ge
(\Delta/r^2) b \Lambda r^{-2}  +b(\Delta/r^2) r^{-3},
\end{equation}
whereas, for all $r>r_+$, we have
\begin{equation}
\label{Vpos2}
|V'(r^*)|\le 
B(\Delta/r^2) \Lambda r^{-3}  +B(\Delta/r^2)r^{-4}.
\end{equation}
Finally, again choosing $a_s$ sufficiently small, 
it follows that for all $r\le r_{c}$ 
\begin{equation}
\label{andforlarger0}
V'(r^*)\ge b(\Delta/r^2) \Lambda,
\end{equation}
whereas there exists a constant $R_c$ such that for $r\ge R_c$,
\begin{equation}
\label{andforlarger}
-V'(r^*) \ge (\Delta/r^2) b \Lambda r^{-3}  +b(\Delta/r^2) r^{-4}
\end{equation}
in this frequency range. 

We will now
split the frequency range $\mathcal{F}_{\mbox{$\flat$}}$ into two subcases,
considering each separately.

\subsubsection{The subrange $|\omega|\le \omega_3$ (the near-stationary subcase)}
\label{oneofthesub}
The motivation for the current to be constructed here
is that in the Schwarzschild or $m=0$ case,
applying $Q_1^h$ with $h=1$
immediately excludes nontrivial stationary solutions $\omega=0$.

We will fix an $\omega_3>0$\index{fixed parameters! small parameters! $\omega_3$ (frequency parameter)}
which will be constrained in this subsection
to be small. Because $\omega_3$ is not exactly $0$ the naive current $h=1$ must
be modified.

We begin with a $Q_1^h$ 
current  which will be defined with an $h=h(r^*)$.

Note first that  given arbitrary $a_0<M$ and $q>0$\index{fixed parameters! 
small parameters! $q$ (associated
with $\mathcal{F}_{\mbox{$\flat$}}$)}, $p>0$\index{fixed parameters! small
parameters! $p$ (associated
with $\mathcal{F}_{\mbox{$\flat$}}$)}, $R_{c}>0$\index{fixed parameters!
large parameters! $R_c$ (large parameter associated
with $\mathcal{F}_{\mbox{$\flat$}}$)}
such that $e^{-p^{-1}}R_c$ is sufficiently large
and $p$ sufficiently small, for each $|a|\le a_s$,
we can define a function $h(r^*)$, such that the following hold:
For
$r\le r_{c}$,
\[
0\le h\le R_c^{-2}, \qquad |h''(r^*)|\le q/|r^*|^2,
\]
with $h(r^*)$ moreover of compact support when restricted to $-\infty<r^*\le r_{c}^*$,
whereas
for $r_{c}\le r\le e^{-p^{-1}}R_c$, 
\[
h\ge  R_c^{-2} \Delta/r^2, \qquad h''(r^*)\le 0,
\]
whereas
for $e^{-p^{-1}} R_c\le r \le R_c$,
\[
|h'(r^*)| \le 4R_{c}^{-2} p /r,
\qquad
|h''(r^*)| \le 4R_{c}^{-2} p/r^2,
\]
whereas, finally, for $r\ge R_{c}$,
\[
h=0.
\]
This $h$ will be useful in view of the positivity of $V_+$. 
We note finally that $h$ can be chosen so that $h$ restricted to $r\ge r_{c}$
is independent of the parameter $q$.

Let us also consider a current $Q_2^y$ 
defined with a $y=y(r^*)$. Like $h$, the function $y$
will be independent of $\omega$, $m$, and $\ell$ in the allowed range. 
Given the parameters $R_{c}$, $p$, with $e^{-p^{-1}}R_{3}$ sufficiently
large and $p$ sufficiently small, the function will satisfy the following properties:
We set  $y(r_{c}^*)=0$,
and for $r\le r_{c}$,
\[
y'(r^*)=h,
\]
noting that $y$ is bounded below in $r\le r_c$ by a potentially large negative constant
depending on $q$, in view of the fact that $h$ is identically
$0$ for sufficiently low $r^*$. 
For $r_{c}\le r\le e^{-p^{-1}}R_{c}$,   we require
\[
y'(r^*)\ge 0,\qquad 
(yV)'(r^*)\le \frac12 R_{c}^{-2} \Delta/r^2 ,
\]
whereas for $e^{-p^{-1}} R_{c}\le  r\le R_{c}$
\[
y'(r^*)\ge bR^{-2}_{c},\qquad -(yV)' (r^*)\ge bR^{-2}_{c}/r^2,
\]
whereas finally,
for  $r\ge R_{c}$,
\[
y=1.
\]  	
In general, for sufficiently large $e^{-p^{-1}}R_{c}$ 
and small $p$, we can indeed construct such a $y$
in view essentially of $(\ref{Vpos1})$, $(\ref{Vpos2})$ and $(\ref{andforlarger})$.
Let us add that $y$ restricted to $r\ge r_{c}$ can be chosen independently
of the choice of $h$.

Consider now the current $Q=Q_1^h+Q_2^y$. 
In $r\le r_{c}$, recall the one-sided bound
\[
-yV' \ge 0
\]
which follows from $(\ref{andforlarger0})$.
It follows that in $r\le r_{c}$,
\[
Q'(r^*)
 \ge 
-q|r^{*}|^{-2}|u|^2
+h{\rm Re}(u\bar H)+2y{\rm Re}(u'\bar H).
\]
For $r_{c} \le r \le e^{-p^{-1}} R_{c}$, choosing $p$ appropriately, choosing $\omega_3$, $a_s$ sufficiently small
\[
Q'(r^*)
 \ge R_{c}^{-2} (\Delta/r^2) |u'|^2 + b\, V_+ |u|^2   +h{\rm Re}(u\bar H) +2y{\rm Re}(u'\bar H).
\]
For $e^{-p^{-1}}R_c\le r \le R_c$, we have for $p$, $\omega_3$ suitably small 
\[
Q'(r^*)
\ge bR_{c}^{-2} (\Delta/r^2)  |u'|^2  +h{\rm Re}(u\bar H) +2y{\rm Re}(u'\bar H).
\]
Finally, for $r\ge R_{c}$ we have,
\[
Q'(r^*)
=(Q_2^y)'(r^*) \ge 2y{\rm Re}(u'\bar H).
\]
We obtain finally,  
\begin{align}
\notag
b(\lambda_1) \int_{r^*_{c}}^{R_c^*}
&(\Delta/r^2) |u'|^2  + (\Delta/r^2)r^{-3}(1+\Lambda
+\omega^2)|u|^2    \, dr^*\label{eq:q}\\
\le&
\int_{-\infty}^{r_{c}^*}
q|r^{*}|^{-2}|u|^2\, dr^*
 +\int_{-\infty}^{+\infty}\left(2y{\rm Re}(u'\bar H)+h{\rm Re}(\bar Hu) \right)\,dr^* \\
&+Q_2^y({{\infty}})-
(Q_1^h+Q_2^y)
({{-\infty}}).\notag
\end{align}

The above is the prototype of the type of
inequality we derive for all frequencies.
Cf.~$(\ref{modred})$. Note, however, the extra term proportional to $q$. We shall
deal with this in Section~\ref{noteonsum}.

\subsubsection{The subrange $|\omega|\ge \omega_3$ (the non-stationary subcase)}
The construction of this section will yield a positive current
for $\omega_3$  small and
$\omega_1$, $\lambda_1$
large depending on the smallness of $a_s$. 

Let us note that in fact, the
 construction of this section is quite general and depends only on the
(almost) positivity properties of the potential $V$ together with its 
asymptotic properties, and not, in particular, on the sign of $V'$,
which will be crucial in Section~\ref{kloubi}.
We note that similar constructions have a long tradition in spectral theory
and are typically used to prove continuity of the spectrum away from the $\omega=0$.

We consider a $Q_2^y$-current.
Recall that we have
\begin{align}
(Q_2^y)' (r^*)&= y'|u'|^2+\omega^2 y'|u|^2 - (yV)' |u|^2+ y{\text{Re}} (u'\overline H)\\ \label{eq:bl}&=
y'|u'|^2+\omega^2 y'|u|^2 - (yV_+)' |u|^2 - y'V_{\text{new}} |u|^2- yV'_{\text{new}} |u|^2+ y{\text{Re}} (u'\overline H) . 
\end{align}
Note that for general bounded $y\ge 0$, $y'>0$, we have
\begin{eqnarray*}
\int_{{-\infty}}^{\infty} (yV_{+})' |u|^2 &=& -\int_{{-\infty}}^{\infty} 
yV_{+}( u \bar{u}' +u'\bar{u})
 + yV_{+}|u|^2({\infty})-yV_{+}|u|^2({-\infty}) \\
&\le&
 \frac12\int_{{-\infty}}^{\infty} y' |u'|^2  + 
 2\int_{{-\infty}}^{\infty} (y^2V_{+}^2/y')|u|^2\\
 &&\hbox{}+ yV_{+}|u|^2({\infty})-yV_{+}|u|^2({-\infty}) \\
&\le&
\frac12\int_{{-\infty}}^{\infty} y' |u'|^2  + 2\int_{{-\infty}}^{\infty} (y^2V_{+}^2/y')
|u|^2.
\end{eqnarray*}

We now define 
\begin{equation}
\label{hereydef}
y= e^{-2\omega_3^{-1} \int_{r^*}^{\infty} V_{+} dr^*}.
\end{equation}
Note that in this case, we have that
\[
y'=
2\omega_{3}^{-1}V_{+}e^{-2\omega_3^{-1}\int_{r^*}^{\infty} V_{+} dr^*}.
\]
We thus have for this choice of $y$ that
\begin{eqnarray*}
\int_{{-\infty}}^{\infty} (yV_{+})' |u|^2
&\le&
\frac12\int_{{-\infty}}^{\infty} y' |u'|^2  + \frac12
\int_{{-\infty}}^{\infty} \omega_3^2 y' |u|^2.
\end{eqnarray*}
The potential $V_{\text{new}}$ obeys the estimates
$$
|V_{\text{new}}|\le a_s B (\la_1+a_s^2\omega^2_1) r^{-4},\qquad 
|V'_{\text{new}}|\le a_s B \Delta/r^2 (\la_1+a_s^2\omega^2_1) r^{-5}.
$$
Therefore, the term  $\frac12 \omega_3^2 y' |u|^2$ dominates both $-y'V_{\text{new}} |u|^2$ and $- yV'_{\text{new}} |u|^2$
in \eqref{eq:bl}, provided that $a_s\ll \omega_3^2$. Integrating, we finally obtain
the inequality:
\begin{align*}
b(\omega_3, \omega_1, \lambda_1) \int_{r^*_{c}}^{R_c^*}
&(\Delta/r^2) |u'|^2  + (\Delta/r^2)(1+(\Lambda+\omega^2)|u|^2\\
\le &\int_{{-\infty}}^{{\infty}}2y{\rm Re}(u'\bar H)
+Q_2^y({{\infty}})
-Q_2^y({{-\infty}}).
\end{align*}

\subsection{The $\mathcal{F}_{\lessflat}$ range (angular dominated frequencies)}

The considerations of this section will constrain
$\lambda_1$ to be suitably large, depending on the choice of $\omega_1$. This range is 
in fact covered by the arguments in Section~\ref{sec:large} for the general case $|a|<M$,
and, of course, 
in~\cite{dr7}. Briefly,
a simple construction in this regime involves the current $Q=Q_0^f+Q_1^h$ with a 
positive and monotonically increasing $f$ and a function $h$ of compact support peaking 
in the region where $V$ dominates $\omega^2$. The final resulting estimate is:
\begin{align*}
 b \int_{r^*_{c}}^{R_c^*}
 & \left (|u'|^2  + (\Delta/r^2) r^{-3}(1+(\Lambda+\omega^2) |u|^2\right)\\
\le
 &\int_{{-\infty}}^{{\infty}}2f{\rm Re}(u'\bar H)+(f'+h){\rm Re}(\bar Hu)\\
&+{Q_0^f}({{\infty}})-{Q_0^f}({{-\infty}}).
\end{align*}

\subsection{The $\mathcal{F}_{\mbox{$\natural$}}$ range (trapped frequencies)}
\label{kloubi}
This is the frequency range of trapping. It is once again subsumed by the 
trapped case $\mathcal{G}_{\mbox{$\natural$}}$ which is treated
in detail in Section~\ref{sec:large}.  

Thus, we take the opportunity here to merely 
highlight the simpler construction in~\cite{dr7} which is special to the $|a|\ll M$ case.
This is a good warm-up for following the detailed but
more technical construction of $\mathcal{G}_{\mbox{$\natural$}}$ in Section~\ref{thisissparta}.

The main observation is that 
$V'$ has a {\bf \emph{unique non-degenerate zero}}
$r_{m\ell}^{(a\omega)}$\index{fixed parameters! $r$-parameters!
$r_{m\ell}^{(a\omega)}$ (unique zero of $V'$ for
$\mathcal{F}_{\mbox{$\natural$}}$ range)} for 
$(\omega, m, \ell)\in\mathcal{F}_{\mbox{$\natural$}}$, depending smoothly on
$\omega$.
\emph{This can be thought of as a quantised version of the property of
trapped null orbits described at the end of Section~\ref{TNG}.}

We consider the current $Q_0^f$:
$$
(Q^f_0)'=2f'|u'|^2-fV' |u|^2 -\frac 12 f{'''}|u|^2+2 f \,{\text{Re}} (u'\overline{H}) + f'\, {\text{Re}} (u\overline{H}).
$$
The function $f$ is chosen in such a way that 
\begin{enumerate}
\item $f'(r^*)\ge 0$,
\item $f< 0$ for $r< r_{m\ell}^{(a\omega)}$ and $f> 0$ for $r> r_{m\ell}^{(a\omega)}$,
\item
\label{proper3}
 $f'''<0$ in the neighborhood of $r_{m\ell}^{(a\omega)}$,
 \item
 $-fV'>0$ dominates outside of a small neighborhood $r_{m\ell}^{(a\omega)}$.
 \end{enumerate}
We obtain finally
\begin{align*}
b\int_{r^*_{c}}^{R_c^*}&\left (
(\Delta/r^2)r^{-2}
|u'|^2  +  (\Delta/r^2)r^{-2}|u|^2+ 
\Delta/r^2 r^{-5} (r-r_{m\ell}^{(a\omega)})^2 (\Lambda+\omega^2) |u|^2\right) \\
\le
 &\int_{{-\infty}}^{{\infty}}2f{\rm Re}(u'\bar H)+f'{\rm Re}(\bar Hu)
+Q_0^f({\infty})-Q_0^f({-\infty}).
\end{align*}

\subsection{The $\mathcal{F}_{\mbox{$\sharp$}}$ range (time-dominated frequencies)}
\label{sharpest!}
Again, this range is subsumed by the $\mathcal{G}_{\mbox{$\sharp$}}$ regime in the general 
$|a|<M$ treated in detail in Section~\ref{sec:large}, so here we can be brief.

For small enough $\lambda_2$ and large enough $\omega_1$, 
we have
\[
\omega^2-V \ge \frac{1-c}2\omega^2,
\]
for frequencies
in $\mathcal{F}_{\mbox{$\sharp$}}$. We apply the current $Q_2^y$ with a function
$y$ chosen so that
$y\ge 0$, $y'\ge B(\Delta/r^2) r^{-2}$
in $(r_{+},R_c]$ and $y=1$ for $y\ge R_c+1$. Since, for 
$(\omega, m, \ell)\in  \mathcal{F}_{\mbox{$\sharp$}}$, we have
\[
|V'| \le B(\Delta/r^2)r^{-3}((\lambda_{m\ell}+a_0^2\omega^2)+1).
\]
We obtain
\begin{align*}
b\int_{r^*_{c}}^{R_c^*}
 &\Delta/r^2 r^{-2}\left(|u'|^2  +(\omega^2 + \Lambda)|u|^2\right)\\
\le&\int_{{-\infty}}^{{\infty}}2y{\rm Re}(u'\bar H)
+Q_2^y ({{\infty}})-
Q_2^y({{-\infty}}).
\end{align*}
See Section~\ref{sec:large} (or alternatively~\cite{dr7}) for more details.

\subsection{A note on summation and the red-shift}
\label{noteonsum}
We have obtained inequalities $(\ref{modred})$ for all frequency triples, 
except for an extra term
\[
\int_{-\infty}^{r_{c}^*}
q|r^{*}|^{-2}|u|^2\, dr^*
\]
coming from $(\ref{eq:q})$ corresponding to the $|\omega|\le \omega_3$
sub-case of the $\mathcal{F}_{\mbox{$\flat$}}$ frequency range.

After summing, we thus in fact obtain:
\begin{align}\label{eq:sum}
b \int_0^\infty d\tau \int_{\Sigma_\tau\cap \{r\in [R_0,R_1]\}}  \left (\chi |\pa\psi|^2 + |\psi|^2\right)&\le
\int_{\Sigma_0} |\pa\psi|^2 + q \int_0^\infty d\tau \int_{\Sigma_\tau\cap \{r\le r_c\}} |\pa\psi|^2 \\ &+
 \int_{{\mathcal H}^+} |{\bf J}^T_\mu n^\mu_{{\mathcal H}}| +
 \int_{{\mathcal I}^+} {\bf J}^T_\mu n^\mu_{{\mathcal I}} \notag
\end{align}

We now add the red-shift estimate \eqref{eq:rshift} and the ${\bf J}^T$ identity \eqref{eq:add} to eliminate 
the small $q$-term and the ${\mathcal H}^+, {\mathcal I}^+$ contributions. We obtain
again (cf.~$(\ref{komble})$)
$$
b \int_0^\infty d\tau \int_{\Sigma_\tau\cap \{r\in [R_0,R_1]\}}  \left (\chi |\pa\psi|^2 + |\psi|^2\right)\le
\int_{\Sigma_0} |\pa\psi|^2. 
$$
This is the desired integrated local energy decay estimate, which may now be coupled with the red-shift 
estimate and the estimates at null infinity \eqref{eq:newmeth} of section \ref{sec:three} to obtain the statements
of decay for the non-degenerate energy and pointwise decay for the solution of Theorem \ref{thrm:state}.

\section{The general $|a|<M$ case}\label{sec:large}

In Section~\ref{sum&comp}, we have reduced the problem to obtaining
the inequality $(\ref{closedredu})$ for all frequency ranges. 
In this section, we shall obtain precisely these inequalities, completing thus the proof.

We begin with a closer examination of some properties of the potential $V=
V^{(a\omega)}_{m\ell}(r)$.
Without loss of generality we may assume that the specific angular momentum $a$ is positive 
$0<a<M$, as the problem admits a symmetry with respect to simultaneous reversal of
the signs of $a$ and the azimuthal frequency $m$.

\subsection{Properties of the potential $V$}
We decompose the potential $V=V_0+V_1$
with 
\begin{align*}
&V_0=\frac{4Mram\omega -a^2m^2+\Delta(\lambda_{m\ell}+\omega^2a^2)}{(r^2+a^2)^2},\\
&V_1=\frac{\Delta(3r^2-4Mr+a^2)}{(r^2+a^2)^3} -\frac{3\Delta^2 r^2}{(r^2+a^2)^4}.
\end{align*}
We note that $V_1$ does not contain any frequency parameters and that
$$
V_1=\frac {\Delta}{(r^2+a^2)^4}\left[a^2\Delta+2Mr(r^2-a^2)\right]\ge 0.
$$
\begin{lemma}\label{lem:1}
For any values of the parameters $\Lambda> 0$, $m$ and $\omega$
the potential function $V_0$ on the interval 
$(r_+,\infty)$ has at most one maximum and one minimum, achieved at the respective points
$r^0_{max}$ and $r^0_{min}$ with the property that $r^0_{min}<r^0_{max}$.
Moreover, for all sufficiently large $\Lambda$ the value $r^0_{max}$ is bounded uniformly from above
provided that either $m\omega\ge 0$ or $a^2\omega^2\le C \Lambda$ for some constant $C$, in which case
the bound for $r^0_{max}$ may depend on $C$.

\end{lemma}
\begin{proof}
We compute 
\begin{align*}
\frac {d}{dr} V_0&=4maM\omega \left [\frac 1{(r^2+a^2)^2} - \frac {4r^2}{(r^2+a^2)^3}\right ]
+\frac {4ra^2m^2}{(r^2+a^2)^3} + \frac \Lambda{(r^2+a^2)^2} \left [2(r-M) -\frac {4r\Delta}{r^2+a^2}\right]\\
&=\frac 1{(r^2+a^2)^3} \left [4maM\omega(-3r^2+a^2) + 4ra^2m^2\right]\\ &+\frac \Lambda{(r^2+a^2)^3}
(2r^3+2a^2r-2Mr^2-2Ma^2
-4r^3+8Mr^2-4ra^2)\\& =\frac 1{(r^2+a^2)^3} \left [4maM\omega(-3r^2+a^2) + 4ra^2m^2-2\Lambda(r^3+a^2r-3Mr^2+Ma^2)\right].
\end{align*}
Thus, we have 
\begin{align*}
\frac d{dr} \left [(r^2+a^2)^3 \frac d{dr} V_0\right]&=-24Mam\omega r + 4a^2m^2-2\Lambda(3r^2-6Mr+a^2)\\ &=-6\Lambda\left [r^2-2Mr +4Mr\sigma+\frac{a^2}3-\frac 23 a^2 \frac{m^2}\Lambda\right], 
\end{align*}
where we set 
$$
\sigma=\frac {am\omega}\Lambda.
$$
We now examine the points where $\frac d{dr} \left [(r^2+a^2)^3 \frac d{dr} V_0\right]$ vanishes.
Such points are solutions of the equation
$$
r^2-2Mr(1-2\sigma)+\frac {a^2}3(1-\frac {2m^2}{\Lambda})=0,
$$
let us denote them by:
$$
r_{1,2}=M(1-2\sigma)\pm\sqrt{M^2(1-2\sigma)^2-\frac {a^2}3(1-\frac {2m^2}\Lambda)}.
$$
We recall the horizon is determined by the value $r_+:$\,\,\, $r_+^2-2Mr_++a^2=0$. In particular,
$r_+>M$ for all values of $a<M$.

If $\sigma>0$ then $r_2<M$ and the only point where $\frac d{dr} \left [(r^2+a^2)^3 \frac d{dr} V_0\right]=0$ on the interval $(r_+,\infty)$ is 
$$
r_1=M(1-2\sigma)+\sqrt{M^2(1-2\sigma)^2-\frac {a^2}3(1-\frac {2m^2}\Lambda)}.
$$
In that case we note that the function $(r^2+a^2)^3 \frac d{dr} V_0\to -\infty$ as $r\to \infty$,
as long as $\Lambda>0$. Therefore, $\frac d{dr} V_0$ can vanish at most at two points 
$r^0_{min}$ and $r^0_{max}$, where 
$$\frac{ d^2}{dr^2} V_0(r^0_{min})>0,\qquad 
\frac{ d^2}{dr^2} V_0(r^0_{max})<0
$$ 
and $r^0_{min}<r^0_{max}$.

If $r_1$ is not real then  $\frac d{dr} \left [(r^2+a^2)^3 \frac d{dr} V_0\right]$ remains negative
for all $r\ge r_+$, which means that $\frac d{dr} V_0$ can only vanish at one point, where
$V_0$ has a maximum.

If $\sigma<0$ then we must examine 
$$
r_2=M(1-2\sigma)-\sqrt{M^2(1-2\sigma)^2-\frac {a^2}3(1-\frac {2m^2}\Lambda)}=
M(1-2\sigma)\left [1-\sqrt{1-\frac {a^2(1-\frac {2m^2}\Lambda)}{3M^2(1-2\sigma)^2}}\right].
$$
Since $a<M$ and $\sigma<0$, we have
$$
\frac {a^2(1-\frac {2m^2}\Lambda)}{3M^2(1-2\sigma)^2}<\frac 13.
$$
We note that for $0\le x<\frac 13$,
$$
\sqrt{1-x}> 1-\frac{2x}3,
$$
which implies that 
$$
r_2<\frac{2M(1-2\sigma)a^2(1-\frac{2m^2}\Lambda)}{9M^2(1-2\sigma)^2}=
\frac {2a^2(1-\frac{2m^2}\Lambda)}{9M(1-2\sigma)}<\frac 29 M.
$$
This now implies that $r_1$ is the only zero of  $\frac d{dr} \left [(r^2+a^2)^3 \frac d{dr} V_0\right]$
on the interval $[r_+,\infty)$ and the previous argument applies.

The last statement of the lemma easily follows from observing that for all sufficiently large $\Lambda$
the leading order behavior of $ \left [(r^2+a^2)^3 \frac d{dr} V_0\right]$ as $r\to \infty$ is determined by the 
expression
$$
(6\La M -12 Mam\omega)r^2-2\Lambda r^3.
$$
\end{proof}
The next statement effectively establishes that even if $r_{min}$ exists, it can only be `trapped' 
for the value $\omega=\omega_+$.
\begin{lemma}\label{lem:2}
For all values of the parameters $\Lambda$,$\omega$,$m$,
we have that on the horizon
$$
\omega^2\ge V(r_+)
$$
with equality achieved only for $\omega=\omega_+$.
In particular, in the notations of the previous lemma, this implies that 
$$
\omega^2>V_0(r_{min}).
$$
\end{lemma}
\begin{proof}
We compute 
\begin{align*}
\omega^2-V(r_+)&=\omega^2-\frac {4Mr_+am\omega-a^2m^2}{(r_+^2+a^2)^2}\\ &=
\frac 1{4M^2r_+^2}\left [4M^2r_+^2\omega-4Mr_+am\omega+a^2m^2\right]=
\frac {(2Mr_+\omega-am)^2}{4M^2r_+^2}.
\end{align*}
\end{proof}
The equality
$$
\omega^2=V(r_+)
$$
occurs precisely at the threshold of the super-radiance condition
$$
\omega=\omega_+=\frac {am}{2Mr_+}.
$$
We now argue that in the super-radiant regime $V_0$ has only a maximum, that is
the point $r_{min}$ is absent.
\begin{lemma}\label{lem:3}
Let 
$$
m\omega\le \frac {am^2}{2Mr_+}.
$$
Then
$$
\frac {d}{dr} V(r_+)\ge \frac {d}{dr} V_0(r_+)\ge c\Lambda>0.
$$
Since we have already shown that $r^0_{min}$, if exists, has the property that $r^0_{min}<r^0_{max}$
this implies that $r^0_{min}$ does not exist and the potential $V_0$ has its unique critical point at $r^0_{max}$.
\end{lemma}
\begin{remark}\label{rem:freq}
The conclusion of the lemma can be easily extended to the range 
$$
am\omega\le \frac {a^2m^2}{2Mr_+}+\alpha\Lambda
$$
for a sufficiently small constant $\alpha$. The smallness of $\alpha$ depends on the smallness
of the quantity $(M-a)$. 
\end{remark}
\begin{proof}
We compute
\begin{align*}
\frac d{dr} V_0(r_+)&=\frac{4maM\omega}{(r_+^2+a^2)^3}  \left [-3r_+^2+a^2\right ]
+\frac {4r_+a^2m^2}{(r_+^2+a^2)^3} + \frac {2(r_+-M)\Lambda}{(r_+^2+a^2)^2} \\ &=
\frac 1{(r_+^2+a^2)^3}\left [4maM\omega [-3r_+^2+a^2] + 4r_+a^2m^2+2(r_+^2+a^2)(r_+-M)\Lambda\right]
\end{align*}
For the values of the parameters $m\omega<0$ the conclusion of the lemma is obvious. Otherwise,
using the condition
$$
m\omega\le \frac {am^2}{2Mr_+}
$$
we obtain
\begin{align*}
(r_+^2+a^2)^3\frac d{dr} V_0(r_+)&\ge \left [\frac {2a^2m^2}{r_+} (-3r_+^2+a^2) +4r_+a^2m^2+2(r_+^2+a^2)(r_+-M)\Lambda\right]\\ &= \left[\frac {2a^2m^2}{r_+} (-r_+^2+a^2) +2(r_+^2+a^2)(r_+-M)\Lambda\right]\\ &=2(r_+-M)\left [\Lambda (r_+^2+a^2)-2a^2m^2\right]\\&=
4(r_+-M)\left [\Lambda Mr_+-a^2m^2\right].
\end{align*}
The inequalities $\Lambda\ge m^2$ and $r_+>M>a$ imply that $\frac d{dr} V_0(r_+)>0$. 
We finish the proof by observing that 
$$
\frac {d}{dr} V_1(r_+)=\frac {4Mr_+(r_+-M)(r_+^2-a^2)}{(r_+^2+a^2)^4}>0.
$$
\end{proof}

Our final result of this section gives a mathematical embodiment to the miracle which is
key to the validity of Theorem~\ref{thrm:state}, namely, the fact that 
{\bf \emph{superradiant frequencies are not `trapped'.}}
\begin{lemma}\label{lem:4}
For the values $0\le a<M$ and 
$$
0\le m\omega\le \frac {am^2}{2Mr_+}.
$$
we have 
$$
\omega^2<V_0(r^0_{max}).
$$
\end{lemma}
\begin{proof}
Let $r_0\ge r_+$ be such that
$$
m\omega= \frac {am^2}{2Mr_0}.
$$
We compute 
\begin{align*}
\omega^2-V_0(r_0)&=\omega^2-\frac {4Mr_0am\omega-a^2m^2+\Delta\Lambda}{(r_0^2+a^2)^2}\\ &=
\frac 1{(r_0^2+a^2)^2}\left [(r_0^2+a^2)^2\omega^2-4Mr_0am\omega+a^2m^2-\Delta\Lambda\right]\\ &=
\frac 1{(r_0^2+a^2)^2}\left [4M^2r_0^2\omega^2-4Mr_0am\omega+a^2m^2+\omega^2\left ((r_0^2+a^2)^2-4Mr_0^2\right)-\Delta\Lambda\right]\\&=\frac {\omega^2(r_0^2-2Mr_0+a^2)(r_0^2+2Mr_0+a^2)-\Delta\Lambda}{(r_0^2+a^2)^2}\\ &=\frac \Delta{(r_0^2+a^2)^2}\left[\frac {a^2 m^2}{4M^2} (1+\frac{2M}{r_0}+\frac{a^2}{r_0^2})-\Lambda\right].
\end{align*}
We now recall that $a<M<r_0$ and that $\Lambda\ge m(m+1)$ to conclude that 
$$
\omega^2-V_0(r_0)< 0.
$$
\end{proof}

\subsection{The frequency ranges}
\label{freerange}
Let $\omega_1$\index{fixed parameters! small parameters! $\omega_1$ (frequency parameter)}, 
$\lambda_1$\index{fixed parameters! large parameters! $\lambda_1$ (frequency parameter)}
be (potentially large) parameters
to be determined, and $\lambda_2$\index{fixed parameters! small parameters! $\lambda_2$ (frequency parameter)}
be a (potentially small) parameter to be determined.
We define the frequency ranges
$\mathcal{G}_{\mbox{$\flat$}}$\index{frequency ranges! $\mathcal{G}_{\mbox{$\flat$}}$ (frequency triplet range,
`bounded' frequencies)}, 
$\mathcal{G}_{\lessflat}$\index{frequency ranges! $\mathcal{G}_{\lessflat}$ (frequency triplet range,
`angular-dominated' frequencies)},
$\mathcal{G}_{\mbox{$\natural$}}$\index{frequency ranges! $\mathcal{G}_{\mbox{$\natural$}}$ (frequency triplet range, `trapped' frequencies)},
 $\mathcal{G}_{\mbox{$\sharp$}}$\index{frequency ranges! $\mathcal{G}_{\mbox{$\natural$}}$ (frequency triplet range, `time-dominated' frequencies)},  $\mathcal{G}^{\mbox{$\sharp$}}$ 
 by
\begin{itemize}
\item
$\mathcal{G}^{\mbox{$\sharp$}}=\{(\omega, m, \ell)$ :
$\omega^2+\Lambda \ge \lambda_1$, $m\omega\in [0,\frac {am^2}{2Mr_+}]\}$,
\item
$\mathcal{G}_{\mbox{$\sharp$}}=\{(\omega, m, \ell)$ :
$|\omega|\ge \omega_1$, $\Lambda<\lambda_2 \omega^2, m\omega\not\in [0,\frac {am^2}{2Mr_+}]\}$,
\item
$\mathcal{G}_{\mbox{$\lessflat$}}=\{(\omega, m, \ell)$ :
$|\omega|\ge \omega_1$, $\lambda_2 \Lambda > \omega^2, m\omega\not\in [0,\frac {am^2}{2Mr_+}]\}$,
\item
$\mathcal{G}_{\mbox{$\natural$}}=\{(\omega, m, \ell)$ :
$|\omega|\ge \omega_1$, $\lambda_2 \Lambda \le \omega^2\le \lambda^{-1}_2 \Lambda, m\omega\not\in [0,\frac {am^2}{2Mr_+}] \}$,
\item
$\mathcal{G}_{\mbox{$\flat$}}=\{(\omega, m, \ell)$ :
$|\omega|\le \omega_1$, $\Lambda \le \lambda_1\}$.
\end{itemize}

\subsection{The $\mathcal{G}^{\mbox{$\sharp$}}$ range} 
This a large frequency super-radiant regime. 
The lower bound $\Lambda\ge m(m+1)$ together with the superradiant condition implies that
in this regime 
$$
\omega^2\le \left (\frac {a}{2Mr_+}\right)^2 \La
$$
and therefore $\Lambda$ is necessarily large. 
In the superradiant range the 
conclusions of Lemma \ref{lem:3} and Lemma \ref{lem:4} apply. In particular, the potential 
$V_0$ has its only critical point--a maximum--at $r^0_{max}$ and $\omega^2<V_0(r^0_{max})$. 
The latter in fact means that 
$$
V_0(r)-\omega^2\ge c \Lambda,\qquad \forall\, r\in [r_{max}-\delta,r_{max}+\delta]
$$
for some positive constant $c$ and $\delta>0$. Note that $\delta$ can be chosen to be 
universal and independent of the size of (large) $\Lambda$. We further observe that
$r_{max}^0$ is uniformly bounded and is away from the value of $r_+$:
$$
c\le (r_{max}-r_+)\le C
$$
with the constants $c, C$ independent of $\Lambda$. We also have that the full potential 
$V=V_0+V_1$ has the property that
$$
\frac {d}{dr} V(r_+)\ge \frac {d}{dr} V_0(r_+)\ge c \Lambda.
$$

We now argue that the properties of $V$ mimic precisely those of $V_0$ with the value 
$r_{max}^0$ replaced by $r_{max}$: 
$$
|r_{max}-r_{max}^0|\le c \Lambda^{-1}.
$$
In the regime under consideration, the second derivative 
$$
\frac {d}{dr} \left ((r^2+a^2)^3\frac {d}{dr} V_0(r)\right)
$$
may only change its sign once at a point $r_1<r_{max}^0$. If it happens then 
$$
 \frac {d}{dr} V_0(r)\ge c \Lambda,
$$
perhaps with a slightly different value of $c$ for all $r\in [r_+,r_1]$ and, since 
$$
\left|\frac {d}{dr} \left ((r^2+a^2)^3\frac {d}{dr} V_0(r)\right)\right|\le C\Lambda,
$$
we have that 
$$
(r_{max}^0-r_1)\ge c.
$$
On the interval $[r_1,\infty)$, we have
$$
\frac {d}{dr} \left ((r^2+a^2)^3\frac {d}{dr} V_0(r)\right)\le 0,
$$
and we can choose a value $r'_1\in (r_1,r_{max}^0)$ such that 
$$
 \frac {d}{dr} V_0(r)\ge \frac c2 \Lambda,\qquad \forall r\in [r_+,r_1']
$$
and 
$$
\frac {d}{dr} \left ((r^2+a^2)^3\frac {d}{dr} V_0(r)\right)\le -c \Lambda r^2,\qquad \forall r\in [r_1',\infty).
$$
On the other hand, the potential $V_1$ obeys 
$$
|V_1|\le C r^{-3},\qquad 
|\frac {d}{dr} V_1(r)|\le C r^{-4},\qquad |\frac {d}{dr} \left ((r^2+a^2)^3\frac {d}{dr} V_1(r)\right)|\le C r.
$$
Therefore, $V$ cannot not have any critical points on $[r_+,r_1']$ and has a unique maximum 
$r_{\max}\in [r_1',\infty)$. The conclusion about the distance between $r_{max}$ and $r_{max}^0$ 
is now equally straightforward. 

We may now deduce the existence of an interval $(r_{max}-\de,r_{max}+\de)$ with the 
property that
$$
V(r)-\omega^2\ge c\Lambda,\qquad \forall r\in (r_{max}-\de,r_{max}+\de)
$$
and 
\begin{equation}\label{eq:Vdeg}
-(r-r_{max}) \frac d{dr} V(r)\ge c(\de) \Lambda \frac {(r-r_{max})^2}{r^{4}}, \qquad \forall r\in [r_+,\infty).
\end{equation}

\vskip1pc
We now proceed to the construction of a suitable current for the regime 
$\mathcal{G}^{\mbox{$\sharp$}}$.
The current will be of the form:
$$
Q=Q_0^f + Q_1^h + \Gamma Q_{\text{red}}^z-BQ_T.
$$
To elucidate the exposition we describe the procedure in stages.

\vskip1pc
1.~We first apply current $Q^f_0$ with $f$ such that $f=-1$ at $r=r_+$, $f=0$ at $r=r_{max}$, 
$f=1$ at $r=\infty$ and $f'(r)>0$ for all $r$, giving
\begin{align}
\label{metastra}
\nonumber
\int_{r_+}^\infty \left [2f'|u'|^2-fV' |u|^2 -\frac 12 f{'''}|u|^2\right ]&=
\left [|u'|^2+(\omega-\omega_+)^2 |u|^2\right]_{r=r_+} + \left [|u'|^2+\omega^2 |u|^2\right]_{r=\infty}
\\ &- \int_{r_+}^\infty \left [2 f \,{\text{Re}} (u'\overline{H}) + f'\, {\text{Re}} (u\overline{H})\right].
\end{align}
We impose additional conditions on the function $f$ to guarantee that the expression 
$$
-fV'-\frac 12 f'''\ge 0
$$
for all values of $r\in[r_+,\infty)$. Since $f$ vanishes at $r=r_{max}$ and $V'$ obeys the
property \eqref{eq:Vdeg}, we can easily arrange 
$$
fV'\ge c \La \frac {\Delta (r-r_{max})^2}{r^7}.
$$
It remains to choose $f$ with the property that $f'''(r)<0$ in a small neighborhood of $r_{max}$ and
$|f'''(r)|\le C \Delta r^{-5}$.

The left hand side of $(\ref{metastra})$ is non-negative yet degenerate at $r=r_{max}$. The condition 
$V(r_{max})-\omega^2\ge c\La$ indicates this regime is non-trapped and the degeneracy 
may be removed with the help of the current $Q^h_1$. The problem however is a lack of 
control of the boundary terms on the right
hand side due to the  super-radiant condition. The latter will be handled with the help of the red-shift current
$Q_{\text{red}}^z$.
\vskip1pc
\vskip1pc
2.~We now add a large multiple of the $Q^h_1$ current with a non-negative function $h$ supported in 
$[r_{max}-\delta,r_{max}+\delta]$ (that is in the region where $V(r)>\omega^2$), equal to one on
 $[r_{max}-\delta/2,r_{max}+\delta/2]$:
\begin{align}
\label{metatastra2}
\nonumber
\int_{r_+}^\infty &\left [(2f'+Ah) |u'|^2+\left (Ah(V-\omega^2)-fV'\right) |u|^2 -\frac 12 (f{'''}+Ah'')|u|^2\right ]\\&=
\nonumber
\left [|u'|^2+(\omega-\omega_+)^2 |u|^2\right]_{r=r_+} + \left [|u'|^2+\omega^2 |u|^2\right]_{r=\infty}
\\ &- \int_{r_+}^\infty \left [2 f \,{\text{Re}} (u'\overline{H}) + (f'+Ah)\, {\text{Re}} (u\overline{H})\right].
\end{align}
The left hand side of $(\ref{metatastra2})$
is now positive and non-degenerate. Moreover, its integrand has the property 
that it is 
$$
\ge A (|u'|^2 + \Lambda |u|^2),\quad \forall\, r\in [r_{max}-\frac{\delta}2,r_{max}+\frac\delta 2].
$$
The size of constant $A$ is limited only by the lower bound on the (large) frequency parameter 
$\Lambda$.
\vskip1pc
3.~We now choose another large constant $\Gamma$ with the property that $\Gamma\ll A$ and add a
$\Gamma$ multiple of the red-shift current $Q^z_{\text{red}}$ with $z=\La \frac {\zeta}{\tilde V}$ and $\zeta=1$
on $[r_+,r_{max}]$ and vanishing for $r>r_{max}+\delta/2$:
\begin{align*}
\Gamma\int_{r_+}^{r_{max}} &\Lambda \frac {V'}{\tilde V^2}
|u'+i(\omega-\omega_+)u|^2 - \Gamma\int_{r_{max}}^{r_{max}+\delta/2}\Lambda 
\left [\left (\frac \zeta{\tilde V}\right)'|u'+i(\omega-\omega_+)u|^2 - \zeta'|u|^2\right] \\ +
\int_{r_+}^\infty &\left [(2f'+Ah) |u'|^2+\left (Ah(V-\omega^2)-fV'\right) |u|^2 -\frac 12 (f{'''}+Ah'')|u|^2\right ]\\&=
\left [|u'|^2+\left ((\omega-\omega_+)^2-\Gamma\Lambda\right) |u|^2\right]_{r=r_+} + \left [|u'|^2+\omega^2 |u|^2\right]_{r=\infty}
\\ &- \int_{r_+}^\infty \left [2 (f-\Lambda\frac\zeta{\tilde V}) \,{\text{Re}} (u'\overline{H}) + (f'+Ah)\, {\text{Re}} (u\overline{H})\right].
\end{align*}
The first term on the left hand side above is positive. To estimate the second term we note that we 
may assume that on
$[r_{max},r_{max}+\delta/2]$, the potential $\tilde V$ has a lower bound $c\Lambda$. Therefore,
its integrand can be estimated by
$$
\le \Gamma C \left [|u'|^2+ (\omega^2+m^2) |u|^2\right]\ll A\left [|u'|^2+ \Lambda |u|^2\right],
$$
where the replacement of $\omega^2+m^2$ by $\Lambda$ is justified by the superradiant condition.
Therefore the left hand side is positive definite. On the other hand, using the boundary conditions
\eqref{eq:b-}, the boundary term at $r=r_+$ 
on the right hand side is 
$$
\left (2(\omega-\omega_+)^2-\Gamma \Lambda\right) |u|^2\le -\frac{\Gamma}2\Lambda |u|^2
$$
and thus now has a ``good" sign and a large multiple constant. The other boundary term at $r=\infty$
still has the ``wrong" sign but it is merely
$$
2\omega^2 |u|^2
$$
and therefore can be eliminated by a subtraction of a $B$ multiple of the $Q_T$ current
$$
B \omega^2 |u|^2|_{r=\infty}= - B \omega(\omega-\omega_+) |u|^2|_{r=r_+}+
B\omega \int_{r_+}^\infty {\text{Im}} (u\overline H), 
$$
which will create a {\it bounded}
negative term at $r=r_+$. The latter in turn is easily dominated by $\frac{\Gamma}2\Lambda |u|^2$,
provided that $B\ll\Gamma$.
We can summarize the results obtained in this regime in the following inequality:
\begin{align*}
b\left (\int_{r_+}^\infty \frac{\Delta}{r^4}\right. &\left.\left [ |u'|^2+\frac {(\omega^2+\La)}{r} |u|^2\right]+
\int_{r_+}^{r_{max}}\frac {|u'+i(\omega-\omega_+) u|^2}{r-r_+} + (\omega^2+\La) \left [|u|^2|_{r=r_{max}}
+ |u|^2|_{r=r_\infty}\right]\right )\\ &= \int_{r_+}^\infty \left [-2 (f-\Lambda\frac\zeta{\tilde V}) \,{\text{Re}} (u'\overline{H}) - (f'+Ah)\, {\text{Re}} (u\overline{H})+B\omega {\text{Im}} (u\overline H)\right ].
\end{align*}
To end this section we observe  that according to Remark \ref{rem:freq} the frequency range $\mathcal{G}^{\mbox{$\sharp$}}$
can be slightly extended to 
$$
0\le m\omega\le \frac {am^2}{2Mr_+}+\a \Lambda
$$
for some small constant $\Lambda$. This observation, in particular, will allow us to restrict the range 
$\mathcal{G}_{\mbox{$\natural$}}$ to 
$$
m\omega\not\in[0, \frac {am^2}{2Mr_+}+\a \Lambda].
$$
\label{sub:sharp}

\subsection{The $\mathcal{G}_{\mbox{$\sharp$}}$ range} This is a large frequency 
non-superradiant regime in which $\omega$
dominates all the other frequencies, that is $\Lambda$ and consequently $m$. Since this a 
non-superradiant range, the boundary terms at $r=r_+$ and $r=\infty$ of all currents
are dominated by the (positive definite) boundary terms of the $Q_T$ current.\\

In this regime, we shall apply a current 
$$
Q=Q_0^f-Q_1^{h} - B Q_T
$$
with a monotonically increasing positive $f$
satisfying $1\le f\le 2$, and the function $h=f'$. 
We obtain 
\begin{align*}
\int_{r_+}^\infty &\left [f' |u'|^2+\left (f'(\omega^2-V)-fV'\right) |u|^2 \right ]\\&
-\left [|u'|^2-\left((\omega-\omega_+)^2+B\omega(\omega-\omega_+)\right) |u|^2\right]_{r=r_+} - \left [|u'|^2+(1-B)\omega^2 |u|^2\right]_{r=\infty}
\\ &- \int_{r_+}^\infty \left [2 f \,{\text{Re}} (u'\overline{H})+B\omega {\text{Im}} (u\overline H) \right].
\end{align*}
Given that $\omega^2\ge C\Lambda$ with a sufficiently large constant $C$ and 
$$
|V|\le C_1\Lambda + c_1 \omega^2,\quad |V'|\le \frac{\Delta}{r^5} 
(C_1\Lambda + c_1 \omega^2)
$$
for some {\it universal} constants (large  $C_1$ and small $c_1$, which can be chosen to be arbitrarily small depending
on th size of $C_1$ and a lower bound on $|\omega|\ge\omega_1$)
), we may choose $C$ so that 
the integrand on the left hand side is positive and dominates 
$$
c_1 f' \left [|u'|^2+(\omega^2+\La) |u|^2\right].
$$
The boundary terms are positive due to the boundary conditions $u'=-i\omega(\omega-\omega_+)$ 
at $r=r_+$ and $u'=i\omega$ at $r=\infty$, the non-superradiance condition, and a choice of a sufficiently
large constant $B$.
Therefore, with an appropriate choice of $f$
\begin{align*}
b \left (\int_{r_+}^\infty \frac{\Delta}{r^4}\right. &\left.\left [ |u'|^2+{(\omega^2+\La)} |u|^2\right]+
 (\omega^2+\La) \left [|u|^2|_{r=r_{max}}
+ |u|^2|_{r=r_\infty}\right ]\right )\\ & \le
\int_{r_+}^\infty \left [-2 f \,{\text{Re}} (u'\overline{H}) +B\omega {\text{Im}} (u\overline H)\right ].
\end{align*}

\subsection{The $\mathcal{G}_{\mbox{$\natural$}}$ range}
\label{thisissparta}
 This is a non-superradiant regime in which $\omega^2$ 
is comparable to $\Lambda$. This range corresponds to the trapping regime.
For these frequencies potential $V_0$ may have at most two critical points
$r^0_{min}<r^0_{max}$. 

According to the Remark \ref{rem:freq} and the discussion at the end of section \ref{sub:sharp}  we may in fact assume that $m\omega\not\in[0,\frac {am^2}{2Mr_+}+\alpha \La]$
for some small constant $\alpha$. Under this condition and since $\omega\sim \La$,
$$
\omega^2-V(r_+)=\omega^2-V_0(r_+)\ge c \Lambda.
$$
We define $r_0$ to be the largest value with the property that 
for all $r\in [r_+,r_0)$
$$
V_0(r)\le V_0(r_+)+\frac {c}2\La.
$$
If $r_0$ is finite then, since $r_{min}^0<r_{max}^0$, it is clear that if $r_{min}^0$ exists then
$$
r_{min}^0<r_0\le r_{\max}^0.
$$ 
Moreover, in that case, since $V_0(r)\to 0$ as $r\to \infty$, the potential $V_0$ does have 
a maximum point $r_{max}^0$ and the value $r_{max}^0$ is uniformly (in $\Lambda$) bounded
from above with a bound dependent on $\lambda_2$. On the other hand, since 
$|\frac d{dr} V_0(r)|\le C \La r^{-3}$, the value $r_0$, and thus $r_{max}^0$,
is uniformly bounded away from $r_+$. 

We now consider the value $r_1$, where the function
$$
\frac {d}{dr} \left [(r^2+a^2)^3 \frac d{dr} V_0(r)\right]=-6\Lambda\left [r^2-2Mr +4Mr\sigma+\frac{a^2}3-\frac 23 a^2 \frac{m^2}\Lambda\right],
$$
changes its sign from positive to negative. If $r_1$ belongs to the interval $[r_+,\infty)$,
then $r_1< r_{max}^0$.
Moreover, since in this regime the parameter $\sigma=am\omega/\Lambda$ is bounded, the value of $r_1$ is uniformly 
bounded from above. We consider two cases based on the value of $V_0(r_1)$. 

If $V_0(r_1)\le V_0(r_+)+\frac {3c}4\La$, then, in view of the fact that $V_0$ has a unique maximum at $r^0_{max}$,
we have that 
$$
V_0(r)\le V_0(r_+)+\frac {3c}4\La,\qquad \forall r\in [r_+,\max(r_0,r_1)].
$$
Moreover, we can find a small constant $\de>0$ such that the above property holds on a slightly larger interval
$$
V_0(r)\le V_0(r_+)+\frac {3c}4\La,\qquad \forall r\in [r_+,\max(r_0,r_1)+\de],
$$
and 
$$
\frac {d}{dr} \left [(r^2+a^2)^3 \frac d{dr} V_0(r)\right]<-c(\de)\La,\qquad\forall  r\in [\max(r_0,r_1)+\de,\infty).
$$
These conclusions are consequences of the simple bound $|\frac d{dr} V_0(r)|\le C \La r^{-3}$ and the fact  
that $\Lambda^{-1} \frac {d}{dr} \left [(r^2+a^2)^3 \frac d{dr} V_0(r)\right]$ is a quadratic polynomial with bounded coefficients
vanishing at the unique point $r_1$ on the interval $[r_+,\infty)$.

If, on the other hand, $V_0(r_1)\ge V_0(r_+)+\frac {3c}4\La$, then we have that $(r_1-r_0)$ is uniformly bounded 
from below by a small positive constant $c_1$, and since $\frac d{dr} V_0(r_0)\ge 0$, 
we can find a value 
$r_0'\in [r_0,r_1]$ such that 
$$
V_0(r)\le  V_0(r_+)+\frac {3c}4\La,\qquad \forall r\in [r_+,r_0']
$$ 
and 
$$
\frac d{dr} V_0(r)\ge c_1 \La,\qquad \forall r\in [r_0',r_1].
$$
Moreover, the last property can be easily extended to a slightly larger interval 
$$
\frac d{dr} V_0(r)\ge c_1 \La,\qquad \forall r\in [r_0',r_1+\de],
$$
so that, in addition 
$$
\frac {d}{dr} \left [(r^2+a^2)^3 \frac d{dr} V_0(r)\right]<-c(\de)\La,\qquad\forall  r\in [r_1+\de,\infty).
$$
In both cases, we may therefore claim the existence of the value $r_3$, uniformly  bounded away from $r_+$, such that 
$$
V_0(r)\le  V_0(r_+)+\frac {3c}4\La\le \omega^2-\frac c4 \La,\qquad \forall r\in [r_+,r_3]
$$
and, such that, for any $r\in [r_3,\infty)$, either 
$$
\frac d{dr} V_0(r)\ge c_1 \La
$$
or 
$$
\frac {d}{dr} \left [(r^2+a^2)^3 \frac d{dr} V_0(r)\right]<-c(\de)\La.
$$
We note that $r_3$ is either uniformly bounded from above or $r_3=\infty$.

Now adding a bounded potential $V_1$ we obtain that for $ r\in [r_+,r_3]$
$$
V(r)\le \omega^2-\frac c4 \La,
$$
while on the interval $[r_3,\infty)$, if non-empty, potential $V$ has a unique non-degenerate maximum at $r_{\max}$,  
$\Lambda^{-1}$-close to $r_{max}^0$, and 
$$
\frac {d^2}{dr^2} V(r_{max})<-c_1 \La,\qquad \frac {d}{dr} V(r_3)>c_1 \La.
$$

In this regime we use a combination 
$$
Q=Q^f_0-Q_2^{y}-B Q_T
$$
of the currents $Q_0, Q_2$ and $Q_T$. 
The current $Q^f_0$ is applied with a function $f$ which vanishes on the horizon, monotonically increases on $[r_3,\infty)$, changes
its sign from negative to positive at 
$r=r_{max}$ and approaches $1$ as $r\to\infty$, and satisfies the additional condition that 
$$
f'''(r)<0,\qquad \forall r\in [r_3,\infty)
$$ 
The second current will be
$Q^y_2$, restricted to  the interval $[r_+,r_3]$, with a monotonically decreasing function
$y$. Finally, we subtract a large multiple of the current $Q_t$.
We obtain for $Q=Q_0^f-Q_2^y-BQ_T$:
\begin{align*}
-\int_{r_+}^{r_3} \left[y'\left (|u'|^2+(\omega^2-V)|u|^2\right)-yV'|u|^2\right]&+ 
\int_{r_+}^\infty \left [2f' |u'|^2-(fV'+\frac 12 f''') |u|^2 \right ]\\&\hskip -5pc+
\left[-f|u'|^2+(B-f)\omega^2 |u|^2\right]_{r=\infty}+\left[B\omega(\omega-\omega_+)-2y(\omega-\omega_+)^2\right] |u|^2|_{r=r_+}
\\ &\hskip -5pc=- \int_{r_+}^\infty \left [2 f \,{\text{Re}} (u'\overline{H})-f' {\text{Re}} (u\overline{H})-B\omega {\text{Im}} (u\overline{H})  \right]+ \int_{r_+}^{r_3} 
\left [y \,{\text{Re}} (u'\overline{H}) \right].
\end{align*}
By the described properties of the potential $V$, the expression $-(fV'+\frac 12 f''')$ is positive on the interval $[r_3,\infty)$. 
On the interval 
$[r_+,r_3]$, we choose a function $y$ so that 
$$
-y'  (\omega^2-V)+y  V'-(fV'+\frac 12 f''')\ge 0.
$$
Since for these values of $r$
$$
(\omega^2-V)\ge c_1 \La,\qquad |V'|\le C\Lambda \frac {\Delta}{r^2}, \qquad |f|+|f'''|\le C\frac {\Delta}{r^2}
$$
for some positive constants $c_1, C$, it suffices to to fulfill the inequality 
$$
-\frac {d}{dr} y\ge y C + C, 
$$
provided that the constructed function $y$ is positive and monotonically decreasing.
The function
$$
y= -C e^{-Cr}  \int_{r}^{r_3}  e^{Cr'} dr'= 1-e^{C(r-r_3)}
$$
satisfies all the above criteria. Finally, a choice of a large constant positive constant $B$ together with 
the non super-radiant condition $m\omega\not\in [0,m\omega_+]$ and the boundary condition 
$u'=i\omega u$ at $r=\infty$ ensure that both boundary terms at $r=r_+$ and $r=\infty$ are positive.
 
 We summarise the results of this regime with the following inequality, which holds under appropriate 
 assumptions on $f$:
 \begin{align*}
b\left (\int_{r_+}^\infty \frac{\Delta}{r^4}\right. &\left. \left [ |u'|^2+ {(\omega^2+\La)} \frac {(r-r_{max})^2}{r^3} |u|^2\right]+
 (\omega^2+\La) \left [|u|^2|_{r=r_{max}}
+ |u|^2|_{r=r_\infty}\right]\right )\\ &\le \int_{r_+}^\infty \left [-2 f \,{\text{Re}} (u'\overline{H}) +f'{\text{Re}} (u\overline H)+B\omega {\text{Im}} (u\overline H)\right] + \int_{r_+}^{r_3} \left [y {\text{Re}} (u'\overline H)\right].
\end{align*}
 
\subsection{The $\mathcal{G}_{\mbox{$\lessflat$}}$ range} This is a non-superradiant 
regime in which the frequency $\omega^2$
is much smaller than $\Lambda$. In particular, by Remark~\ref{rem:freq},
the potential $V_0$ has only 
one critical point--a maximum--at the value $r=r^0_{max}$. Repeating the arguments which we applied to the 
regime $\mathcal{G}^{\mbox{$\sharp$}}$ we also conclude that the potential $V$ has a unique non-degenerate 
critical maximal point at $r_{max}$, which 
is uniformly bounded away from $r_+$ and is bounded from above. Similarly, we also have 
 an interval $(r_{max}-\de,r_{max}+\de)$ with the 
property that
$$
V(r)-\omega^2\ge c\Lambda,\qquad \forall r\in (r_{max}-\de,r_{max}+\de)
$$
and 
$$
(r-r_{max}) \frac d{dr} V(r)\ge c \Lambda \frac {(r-r_{max})^2}{r^4}, \qquad \forall r\in (r_+,\infty).
$$
We repeat the construction of $\mathcal{G}^{\mbox{$\sharp$}}$ with the current $Q=Q_0^f+AQ_1^h$ and the same choice of the
functions $f$ and $h$, which gives 
\begin{align*}
\int_{r_+}^\infty &\left [(2f'+Ah) |u'|^2+\left (Ah(V-\omega^2)-fV'\right) |u|^2 -\frac 12 (f{'''}+Ah'')|u|^2\right ]\\&=
\left [|u'|^2+(\omega-\omega_+)^2 |u|^2\right]_{r=r_+} + \left [|u'|^2+\omega^2 |u|^2\right]_{r=\infty}
\\ &- \int_{r_+}^\infty \left [2 f \,{\text{Re}} (u'\overline{H}) + (f'+Ah)\, {\text{Re}} (u\overline{H})\right],
\end{align*}
with a non-degenerate positive definite integrand on the left hand side. In this case however, due to 
the non super-radiance condition, we may simply subtract a large multiple of the current $Q_T$ to handle the
boundary terms 
\begin{align*}
\int_{r_+}^\infty &\left [(2f'+Ah) |u'|^2+\left (Ah(V-\omega^2)-fV'\right) |u|^2 -\frac 12 (f{'''}+Ah'')|u|^2\right ]\\&+
B\omega(\omega-\omega_+) |u|^2_{r=r_+}+B\omega^2|u|^2|_{r=\infty}\\&=
\left [|u'|^2+(\omega-\omega_+)^2 |u|^2\right]_{r=r_+} + \left [|u'|^2+\omega^2 |u|^2\right]_{r=\infty}
\\ &- \int_{r_+}^\infty \left [2 f \,{\text{Re}} (u'\overline{H}) + (f'+Ah)\, {\text{Re}} (u\overline{H})-B{\text{Im}} (u\overline H)\right].
\end{align*}
We summarise: 
\begin{align*}
b\left (\int_{r_+}^\infty \frac{\Delta}{r^4}\right. &\left.\left [ |u'|^2+\frac {(\omega^2+\La)}{r} |u|^2\right]+ (\omega^2+\La) \left [|u|^2|_{r=r_{max}}
+ |u|^2|_{r=r_\infty}\right]\right )\\ &\le  \int_{r_+}^\infty \left [-2 f \,{\text{Re}} (u'\overline{H}) - (f'+Ah)\, {\text{Re}} (u\overline{H})+B\omega {\text{Im}} (u\overline H)\right].
\end{align*}

\subsection{The $\mathcal{G}_{\mbox{$\flat$}}$ range} \label{whit}
This is a regime in which all frequencies $\omega$, $m$ 
and $\lambda_{m\ell}$ are bounded. Recall that in the small $a$ case,
a general construction for such frequencies was
presented in Section~\ref{elow} of this paper. As shown in~\cite{dr7}, this
construction also holds for general $a<M$ in the case $m=0$. Moreover,
it can be easily adapted to cover more generally non-superradiant modes
in this regime. To handle the remaining bounded superradiant frequencies in this
regime, the argument requires a modification of
Whiting's results on the absence of exponentially growing modes in
\cite{whiting}. 
The details will be given in \cite{drf2}.

We note that since the azimuthal frequency $m$ can be defined without the Fourier transform in 
time, we can define a priori the projection $P_m$  to the $m$'th azimuthal
mode. Clearly,  solutions of the wave equation such that
$P_m\psi=0$ for $1\le |m|\le m_0$ and large $m_0$ are  not
supported at all in the frequency range $\mathcal{G}_{\mbox{$\flat$}}$. 
Thus, for such solutions, this appeal to~\cite{whiting} is unnecessary and
we have given here a completely self-contained discussion of all constructions
necessary for the proof.

\subsection{Summation and the local integrated energy decay}
We have thus obtained the inequalities $(\ref{closedredu})$ 
for all frequencies, as required,
and we may thus apply Section~\ref{sum&comp} to obtain 

\begin{align}\label{eq:sum'}
b \int_0^\infty d\tau \int_{\Sigma_\tau\cap \{r\le r_+\}}  \left (\chi |\pa\psi|^2 +|\partial_r\psi|^2+ |\psi|^2\right)&\le
\int_{\Sigma_0} |\pa\psi|^2.
\end{align} 
The precise form of the estimates of Section~\ref{thisissparta} 
now reveals the nature of the physical-space degeneration 
function $\chi$: The function $\chi$ must vanish
in a neighborhood of the projection of the `trapped set'
$$
\{r: r\in \cap_{L=1}^\infty (\cup_{\ell\ge L} r^{a\omega}_{m\ell})\},
$$ 
where $r^{a\omega}_{m\ell}$ are the maximum values $r_{max}$ of the potential $V$ in the $\mathcal{G}_{\mbox{$\natural$}}$ range. 

Of course, we have really obtained a  stronger estimate than $(\ref{eq:sum'})$, namely,
we obtain the boundedness of the sum of the left hand sides of the 
versions of $(\ref{closedredu})$ 
obtained here, where each microlocalised $\chi$ degenerates
at most at a single point. The degeneration of this
stronger estimate however, as opposed to $(\ref{eq:sum'})$, can thus
only be understood microlocally.

The reader should return to Section~\ref{sum&comp} 
to obtain the remaining statements of Theorem \ref{thrm:state}.

\end{document}